\documentclass[12pt]{article}
\usepackage{amsfonts,amsthm,amsmath,amssymb,mathabx,comment,a4wide}

\newcommand{\mo}[1]{}

\newcommand{\benda}{\tag*{$\Box$}}

\newcommand{\fancya}{\mathcal{A}}

\newtheorem{theorem}{Theorem}[section]
\newtheorem{corollary}[theorem]{Corollary}
\newtheorem{proposition}[theorem]{Proposition}
\newtheorem{prop}[theorem]{Proposition}
\newtheorem{lemma}[theorem]{Lemma}

\theoremstyle{definition}

\newtheorem{example}[theorem]{Example}

\newtheorem{definition}[theorem]{Definition}
\newtheorem{remark}[theorem]{Remark}
\newcommand{\HOM}[1]{\textsc{Hom}(#1)}

\newcommand{\HOMP}[1]{p\textsc{-Hom}(#1)}

\newcommand{\F}{\mathcal{F}}
\newcommand{\T}{\mathcal{T}}
\renewcommand{\P}{\mathcal{P}}
\renewcommand{\L}{\mathcal{L}}
\newcommand{\R}{\mathcal{R}}

\newcommand{\fG}{\mathcal{G}}
\newcommand{\fH}{\mathcal{H}}

\newcommand{\fC}{\mathcal{C}}
\newcommand{\fM}{\mathcal{M}}

\newcommand{\core}{\mathsf{core}}

\newcommand{\minors}{\mathsf{minors}}

\newcommand{\rela}{\mathbf{A}}
\newcommand{\relb}{\mathbf{B}}
\newcommand{\relc}{\mathbf{C}}
\newcommand{\reld}{\mathbf{D}}
\newcommand{\relt}{\mathbf{T}}

\newcommand{\relf}{\mathbf{F}}
\newcommand{\relg}{\mathbf{G}}
\newcommand{\relh}{\mathbf{H}}
\newcommand{\reli}{\mathbf{I}}
\newcommand{\relm}{\mathbf{M}}
\newcommand{\reln}{\mathbf{N}}

\newcommand{\refl}{\mathsf{refl}}

\newcommand{\nats}{\mathbb{N}}

\newcommand{\rt}{\mathrm{root}}
\newcommand{\ar}{\mathrm{ar}}
\newcommand{\gr}{\mathsf{graph}}

\newcommand{\dom}{\mathsf{dom}}
\newcommand{\im}{\mathsf{im}}

\newcommand{\str}{\textup{STR}}

\newcommand{\res}{\upharpoonright}

\newcommand{\vto}{\stackrel{v}{\to}}
\newcommand{\hto}{\stackrel{h}{\to}}

\newcommand{\qr}{\mathsf{qr}}

\newcommand{\hi}{\mathsf{hi}}

\newcommand{\qfred}{\leq^{\textup{qf}}}
\newcommand{\qfeq}{\equiv^{\textup{qf}}}

\renewcommand{\qfred}{\le_{\mathit{qfap}}}
\renewcommand{\qfeq}{\equiv_{\mathit{qfap}}}
\newcommand{\In}{\mathrm{in}}
\newcommand{\pl}{\textit{pl}}

\newcommand{\numb}[1]{\ulcorner #1\urcorner}

\newcommand{\MC}[1]{p\textsc{-MC}(#1)}

\newcommand{\PP}[2]{\textsc{PP}^{#1}[#2]}
\newcommand{\DPP}[2]{\textsc{DPP}^{#1}[#2]}

\begin{document}

\title{\bf One Hierarchy Spawns Another:\\
Graph Deconstructions and 
the Complexity Classification of Conjunctive Queries\footnote{An extended abstract of this work has appeared in \cite{lics}.}}

\author{
Hubie Chen\\
\small Universidad del Pa\'{i}s Vasco\\[-0.5ex]
\small San Sebasti\'an, Spain\\[-0.5ex]
\small IKERBASQUE, Bilbao, Spain \\[-0.5ex]
\small \texttt{hubie.chen@ehu.es}
\and
Moritz M\"uller\\
\small Kurt G\"odel Research Center\\[-0.5ex]
\small Universit\"at Wien, Austria\\[-0.5ex]
\small \texttt{moritz.mueller@univie.ac.at}
}
\date{}
\maketitle

\begin{abstract}
 We study the problem 
of conjunctive query evaluation relative to a class of queries; 
this problem is formulated here
as the relational homomorphism problem relative to a class of structures~$\mathcal{A}$,
wherein each instance must be a pair of structures such that the first
structure is an element of~$\mathcal{A}$.
We present a comprehensive complexity classification of these problems,
which strongly links graph-theoretic properties of~$\mathcal{A}$ to the 
complexity of the corresponding homomorphism problem.
In particular, we define a binary relation on graph classes,
which is a preorder,
and completely describe the resulting hierarchy given by this relation.
This relation is defined in terms of a notion which we call
\emph{graph deconstruction} and which is a variant of the well-known 
notion of tree decomposition.
We then use this  hierarchy of graph classes to infer a complexity hierarchy 
of homomorphism problems which is comprehensive up to
a computationally very weak notion of reduction, namely,
a parameterized version of quantifier-free first-order reduction.
In doing so, we obtain a significantly refined complexity classification
of homomorphism problems,
as well as a unifying, modular, and conceptually clean
treatment of existing complexity classifications.
We then present and develop the theory of 
Ehrenfeucht-Fra\"{i}ss\'{e}-style 
pebble games 
which solve the homomorphism problems where
the cores of the structures in $\mathcal{A}$ have bounded tree depth.
Finally, we use our framework to classify the complexity of 
model checking existential sentences having bounded quantifier rank.

\medskip


\medskip

\end{abstract}

\section{Introduction}

\emph{Conjunctive queries} are basic and heavily studied
database queries, and can be viewed logically as formulas consisting
of a sequence of existentially quantified variables,
followed by a conjunction of atomic formulas on those variables.
Since the 1977 article of Chandra and Merlin~\cite{ChandraMerlin77-optimal},
complexity-theoretic aspects of conjunctive queries
have been a research subject of persistent and enduring
interest which continues to the present day
(as discussed and evidenced, for example, by the 
works~\cite{AbiteboulHullVianu95-foundationsdatabases,KolaitisVardi00-containment,PapadimitriouYannakakis99-database,GottlobLeoneScarcello01-complexityacyclic,GottlobLeoneScarcello02-hypertreedecomposisions,Grohe07-otherside,SchweikardtSchwentickSegoufin09-querylanguages,Marx10-tractablehypergraph}).
In this article, we study
\emph{conjunctive query evaluation}, 
which is 
the problem of evaluating a conjunctive query on a relational structure.
Conjunctive query evaluation is indeed equivalent to a number of 
well-known and well-studied problems, including
the homomorphism problem on relational structures,
the constraint satisfaction problem,
and conjunctive query containment~\cite{ChandraMerlin77-optimal,KolaitisVardi00-containment}.
That this problem appears in many equivalent guises
attests to its fundamental, primal nature, and
this problem has correspondingly been approached and studied
from a wide variety of perspectives and motivations.

As has been eloquently articulated 
in the literature~\cite{PapadimitriouYannakakis99-database},
the employment of classical complexity notions such as
polynomial-time tractability to grade the complexity of
conjunctive query evaluation is not totally satisfactory:
 a typical scenario---for example, in the database context---is 
the evaluation of a relatively short query
on a relatively large structure.
This asymmetry between the two parts of the input 
suggests 
a notion of complexity
wherein one relaxes the dependence on the query.
For example, in measuring time complexity, one might allow
a non-polynomial dependence on the query while enforcing
a polynomial dependence on the structure.
{\em Parameterized complexity theory}~\cite{FlumGrohe06-parameterizedcomplexity,DF} is a comprehensive framework
for studying problems where each instance has an associated parameter,
and arbitrary dependence in the parameter is permitted;
in the query evaluation setting, the query 
(or the query length)
is naturally 
taken to be the parameter.  We use and focus on this viewpoint
in the present article, and hence in this discussion.

Conjunctive query evaluation is known to be computationally intractable
in general, and consequently a recurring theme in the study of 
this problem
is the identification of structural properties of
conjunctive queries that
provide tractability or other computationally desirable behaviors.
A well-studied framework in which to seek such properties is
the family of parameterized homomorphism problems
$\HOMP{\fancya}$, for classes of relational structures $\fancya$:
the problem $\HOMP{\fancya}$
is to decide,
given a relational structure $\rela$ from $\fancya$ 
and another relational structure $\relb$, whether 
there is a homomorphism from $\rela$ to $\relb$;
the parameter here is the first structure $\rela$.
Studying this problem family amounts to studying conjunctive query evaluation
on various classes of conjunctive queries, as
it is a classical fact that each boolean conjunctive query $\phi$ can
be bijectively represented as a structure $\rela$ in such a way that,
for any structure $\relb$, it holds that $\phi$ is true on $\relb$ if
and only if $\rela$ admits a homomorphism to
$\relb$~\cite{ChandraMerlin77-optimal}. 
We will focus on the case where $\fancya$ has bounded arity, and assume this property throughout this discussion.

\subsubsection*{Known classifications.}
An exact description of the tractable problems 
of the form
$\HOMP{\fancya}$ was presented by Grohe~\cite{Grohe07-otherside}.
In particular, a known sufficient condition for fixed-parameter
tractability of $\HOMP{\fancya}$ was that
the \emph{cores}
of $\fancya$ have 
\emph{bounded treewidth}~\cite{DalmauKolaitisVardi02-treewidth};
Grohe completed the picture by showing
that for any class $\fancya$ not satisfying this condition,
the problem $\HOMP{\fancya}$ is W[1]-complete.
Fixed-parameter tractability 
is a parameterized relaxation of polynomial-time tractability,
and W[1]-hardness can be conceived of as a parameterized analog
of NP-hardness.
Intuitively, the \emph{core} of a structure is an 
equivalent structure
of minimal size.

Later, a classification of the tractable $\HOMP{\fancya}$ problems
was presented~\cite{ChenMueller13-fineclass-arxiv}.
This classification is exhaustive 
 up to parameterized logarithmic space
reductions;
parameterized logarithmic space 
(para-L)
relaxes logarithmic space
in a way analogous to that in which
fixed-parameter tractability relaxes polynomial time.
Let $\T$ denote the class of all trees,
let $\P$ denote the class of all paths,
and for a class of structures $\fancya$, let $\fancya^*$ be
the class of structures obtainable by taking a structure
$\rela$ in $\fancya$ and giving each element its own color.
The classification states that each tractable
$\HOMP{\fancya}$ problem is 
para-L equivalent to $\HOMP{\T^*}$, para-L equivalent to $\HOMP{\P^*}$,
or decidable in para-L. 
The properties determining which of the three behaviors occurs
are \emph{bounded pathwidth} and \emph{bounded tree depth},
established graph-theoretical properties.

\subsection{Contributions}

In this article, we present a significantly refined  complexity
classification of the homomorphism problems $\HOMP{\fancya}$ which is
exhaustive up to an extremely simple and computationally weak notion of
reduction based on quantifier-free interpretations  from first-order
logic. Our classification generalizes the just-described known
classifications, and indeed 
our present study provides a uniform, modular, and self-contained
treatment thereof. 

After presenting the classification, we introduce and 
develop the theory of
Ehrenfeucht-Fra\"{i}ss\'{e}-style 
pebble games for solving
the problems lying at the
lower end of our hierarchy; in doing so,
we obtain a characterization of the homomorphism problems
$\HOM{\fancya}$ solvable in classical logarithmic 
space.  As a further technical contribution, we utilize our 
framework to analyze 
the complexity of model checking existential sentences.

We now give an overview of and further details on the contributions and obtained results;  
we refer the reader to the technical sections for precise statements.

\subsubsection*{A graph-theoretic hierarchy.}
Previous work~\cite{ChenMueller13-fineclass-arxiv} 
related the complexity of conjunctive queries
to the named graph-theoretic properties by 
showing that certain relationships on graph classes implied
reductions for the corresponding homomorphism problems, 
for example:
\begin{itemize}\itemsep=0pt

\item If $\fG$ is a graph class and $\fM$ is the class of minors of graphs in $\fG$,
then $\HOMP{\fM^*}$ reduces to $\HOMP{\fG^*}$.

\item 
If the members of a graph class $\fG$ have bounded width tree decompositions whose trees lie in a 
graph class $\fH$, then $\HOMP{\fG^*}$ reduces to $\HOMP{\fH^*}$.

\end{itemize}
Another known and important
reduction~\cite{GroheSchwentickSegoufin01-conjunctivequeries} 
is as follows:
\begin{itemize}\itemsep=0pt

\item When $\R$ is the class of all grids and $\fG$ is any graph
class, then $\HOMP{\fG^*}$ reduces to $\HOMP{\R^*}$.

\end{itemize}

We give a unified explanation for all of these reductions by defining
a binary relation $\preceq$ on graph classes.
This relation has the key property:
\begin{equation}\label{eq:introreduction}
\textup{If $\fG \preceq \fH$, then $\HOMP{\fG^*}$ reduces to $\HOMP{\fH^*}$.}
\end{equation}
This key property is shown to imply the three just-named results.

The definition of the relation $\preceq$ is simple and is based on 
a notion which we call \emph{graph deconstruction} 
and which 
is strongly related to and inspired by the notion of tree decomposition.
When $\relg$ and $\relh$ are graphs,
we define an $\relh$-deconstruction of $\relg$
to be a family $(B_h)_{h \in H}$ of subsets of the vertex set of $\relg$
which is indexed by the vertex set $H$ of $\relh$ and which
satisfies properties similar to those in the
definition of tree decomposition
(Definition~\ref{def:deconstruction});
note that here, the graph $\relh$ is not restricted to be a tree,
as it is in the definition of tree decomposition.
Each such deconstruction $(B_h)_{h \in H}$ 
has associated with it a measure called
its \emph{width} which is based on the sizes of the subsets $B_h$.
For graph classes~$\fG$ and~$\fH$, we define
$\fG \preceq \fH$ if and only if there is a constant $w$
where for each graph~$\relg$ in $\fG$, there is a graph $\relh$ in $\fH$
such that $\relg$ has an $\relh$-deconstruction of width 
at most $w$.

We describe completely the hierarchy that this relation yields
on graph classes.
Define~$\T_n$ to be the class of trees of height at most~$n$;
$\F_n$ to be the class of forests of height at most $n$; and
$\L$ to be the class of all graphs.
We present the following hierarchy: 
\begin{center}
$\T_0 \preceq \F_0 \preceq \T_1 \preceq \F_1 \preceq \cdots
\preceq \P \preceq \T \preceq \L.$
\end{center}
We prove that this hierarchy is strict and \emph{comprehensive} in
that every graph class is equivalent to exactly one of the classes in the hierarchy (Theorem~\ref{thm:graph-hierarchy}).
To understand the upper levels of the hierarchy, we 
use 
known excluded minor theorems.
To determine the lower part of the hierarchy (below $\P$), 
we introduce a new complexity measure on graph classes $\fG$
which we call \emph{stack depth}: 
the maximum $d$ such that all depth $d$ trees are minors of $\fG$; this equals
the minimum depth of forests that allow for
bounded width decompositions of the graphs in $\fG$
(Lemma~\ref{lemma:stack-depth-placement}).
%

\subsubsection*{A complexity-theoretic hierarchy.}
Having understood the 
relation $\preceq$ on graph classes,
we then turn to study homomorphism problems.
As mentioned, we prove that
$\fG \preceq \fH$ implies $\HOMP{\fG^*}$ reduces to $\HOMP{\fH^*}$ (Theorem~\ref{theo:dechom});
this property was indeed a primary motivation for the definition of
$\preceq$.  
We prove this with respect to a computationally weak notion of reduction
that we call \emph{quantifier-free after a precomputation (qfap)}.
This notion of reduction naturally unifies and
incorporates two modes of computation that have long been studied.
For each parameter, this reduction provides 
a quantifier-free, first-order \emph{interpretation}
that defines the output instance
in the input instance;
interpretations as reductions have a tradition in descriptive 
complexity theory~\cite{immerman,EbbinghausFlum95-finitemodeltheory}, and Dawar and He studied 
a type of quantifier-free interpretation
in the parameterized setting~\cite{DawarHe09-logicalreductions}.
In addition, our reduction allows for \emph{precomputation}
 on the parameter of an input instance.
This follows an established schema in the definition of 
parameterized modes of computation: 
fixed-parameter tractability can be defined as polynomial time
after a precomputation, and para-L can be defined as logarithmic space
after a precomputation~\cite{FlumGrohe03-describing}.  

Let $\fancya$ be an arbitrary class of bounded-arity structures. We prove that
$\HOMP{\fancya}$ is equivalent under qfap-reductions
to $\HOMP{\fG^*}$, where $\fG$ is the class of graphs of the cores of
the structures in $\fancya$ (Theorem~\ref{theo:coregraphhom}).
By this theorem, property~\eqref{eq:introreduction} and 
our description of the graph hierarchy, we obtain:
\begin{center}
Every problem of the form $\HOMP{\fancya}$
is equivalent, under qfap-reduction, \\
to a problem $\HOMP{\fH^*}$,
where $\fH$ is a graph class from the graph hierarchy.
\end{center}
This interestingly implies that, with respect to qfap-reduction, the 
complexity degrees attained by problems $\HOMP{\fancya}$ are
linearly ordered according to the classes in our graph hierarchy,
in particular, linearly ordered in a sequence of order type~$\omega +
3$.

In brief, our approach for understanding the family of problems $\HOMP{\fancya}$ is to present a graph hierarchy and then show that this hierarchy induces a complexity hierarchy for these problems. 
\emph{We wish to emphasize the unifying nature, the modularity, and the conceptual cleanliness of this approach.}
Our definition and presentation of the graph hierarchy cleanly and neatly encapsulates the graph-theoretic content needed to present the complexity hierarchy. We obtain a uniform and self-contained derivation of the mentioned known classifications~\cite{GroheSchwentickSegoufin01-conjunctivequeries,Grohe07-otherside,ChenMueller13-fineclass-arxiv}, which derivation we find to be clearer and simpler than those of the original works.
This uniform derivation is a testament to 
the utility of the graph hierarchy,
and we view the simple definitions of 
\emph{graph deconstruction} and
the relation $\preceq$,
as well as the development of their basic theory, 
as conceptual contributions.

Our treatment strengthens the known classifications, since the problems classified as being computationally equivalent are here shown to be so under qfap-reductions; in particular, it follows from our treatment that all of the  W[1]-complete problems from Grohe's theorem~\cite{Grohe07-otherside} are pairwise equivalent under qfap-reductions and hence in  a very strong sense.

\subsubsection*{A proof of Grohe's theorem.}
We present a modular, conceptually concise, and relatively transparent
proof of Grohe's celebrated hardness result
(Section~\ref{sect:grohes-theorem}).
In particular, we make simple use of our notion of
graph deconstruction, and we obtain this hardness result
essentially as the composition of three readily comprehensible
polynomial-time reductions.

\subsubsection*{Consistency, pebble games, and logarithmic space.} 
The study of so-called consistency algorithms has a long tradition
in research on constraint satisfaction and homomorphism problems.
Such algorithms are typically efficient and simple heuristics
that can detect inconsistency (that is, that an instance is a \emph{no}
instance) and are based on local reasoning.
Identifying cases of these problems where
such algorithms provide a sound and complete decision procedure
has been a central theme in the tractability theory of these problems
(see for example~\cite{KolaitisVardi00-gametheoretic,DalmauKolaitisVardi02-treewidth,BartoKozik09-boundedwidth,ChenDalmau05-hypertree,BulatovKrokhinLarose08-dualities,AtseriasWeyer-Consistency,ChenDalmauGrussien13-ACandfriends}).

Seminal work of Koalitits and Vardi~\cite{KolaitisVardi00-gametheoretic} 
showed that certain natural consistency algorithms 
could be viewed as
determining the winner in 
certain Ehrenfeucht-Fra\"{i}ss\'{e} type pebble games~\cite{EbbinghausFlum95-finitemodeltheory}. 
Since this work, there has been sustained effort devoted
to presenting pebble games that solve cases 
of the homomorphism problem.
For example, 
there is a study of
pebble games that solve $\HOMP{\fancya}$ 
when the class $\fG$ of graphs of structures from $\fancya$ has bounded treewidth 
\cite{DalmauKolaitisVardi02-treewidth,AtseriasBulatovDalmau-kconsistency},
and also when $\fG$ has bounded pathwidth~\cite{Dalmau05-linear}.

We complete the picture by presenting natural pebble games
that are shown to solve
$\HOMP{\fancya}$ 
when $\fG$ has bounded tree depth
(Section~\ref{sect:pebble-games}).
Our pebble games are finite-round games that can be thought of as
homomorphism variants of the classical Ehrenfeucht-Fra\"{i}ss\'{e} game.
We develop the theory of our games, showing for example that 
it is decidable, given a structure $\rela$,
whether or not a particular game solves the homomorphism problem
on $\rela$ (Theorem~\ref{thm:char-v-game-solves}).
We also show equality $\pm 1$
 between the number of pebbles needed 
to solve $\HOMP{\fancya}$ 
and the tree depth of $\fG$;
and, 
between the number of rounds needed
to solve $\HOMP{\fancya}$ 
and the stack depth of $\fG$
(Section~\ref{sect:pebble-games}).
We believe that the latter result reinforces the suggestion
that stack depth is a natural graph-theoretic measure.

We obtain a characterization of the classical homomorphism problems
$\HOM{\fancya}$ decidable in classical logarithmic space: these are
precisely those  where the cores of structures from $\fancya$ have
bounded tree depth
(Section~\ref{sect:logspace}).
This characterization is established
under a natural hypothesis from 
parameterized space complexity.

\subsubsection*{Model checking existential sentences.}
The given hierarchies, along with the notion of qfap reduction,
provide a clean and comprehensive understanding of the complexity
degrees of parameterized homomorphism problems.
 We expect that the given hierarchies can be meaningfully used to
 obtain a fine-grained understanding of the  complexity of other
 problems of independent interest.  
We show
that the hierarchy can be used to classify the complexity of model-checking existential sentences having bounded quantifier rank
(Section~\ref{sect:existential-sentences}).

\section{Preliminaries}

For $n\in\nats$ we let $[n]$ denote $\{1,\ldots,n\}$ and understand $[0]=\emptyset$.

\subsection{Structures}

A  \emph{vocabulary} is a finite set of relation symbols,
where each symbol $R$ has an associated arity $\ar(R) \in \nats$.
A {\em structure $\relb$ with vocabulary $\sigma$}, for short, a {\em $\sigma$-structure} 
is given by a non-empty set $B$ called its \emph{universe} together with 
an \emph{interpretation} $R^{\relb}\subseteq B^{\ar(R)}$ of $R$ for every $R \in \sigma$. We only consider finite structures, i.e. structures with finite universe.
When $\relb$ is a $\sigma$-structure and $S$ a non-empty subset of $B$,
we let $\langle S \rangle^{\relb}$ 
denote the $\sigma$-structure {\em induced} in $\relb$ on $S$: it has universe $S$ and
interprets every $R \in \sigma$ by $S^{\ar(R)} \cap R^{\relb}$. The class of all $\sigma$-structures is denoted by $\str[\sigma]$, and the class of all structures 
by~$\str$.
The {\em  product} $\rela\times\relb$ of two $\sigma$-structures $\rela$ and~$\relb$ has
universe $A\times B$ and interprets $R\in\sigma$ by 
$$
R^{\rela\times\relb}:=\big\{((a_1,b_1),\ldots,(a_{\ar(R)},b_{\ar(R)}))\mid \bar a\in R^\rela,
\bar b\in R^\relb\big\}.
$$

For a vocabulary $\sigma$, a (first-order) $\sigma$-formula $\varphi$ is built from {\em atoms} $Rx_1\cdots x_{\ar(R)}$ and $x=y$ where $x,y$ and the $x_i$ are 
variables and $R\in \sigma$, by means of Boolean combinations $\vee,\wedge,\neg$ and existential and universal quantification $\exists x,\forall x$. 
We write $\varphi(\bar x)$ for $\varphi$ to indicate that the free variables of $\varphi$ are among the components of $\bar x=x_1\cdots x_r$;
for a $\sigma$-structure, and $\bar a=a_1\cdots a_r\in A^r$ we write $\rela\models\varphi(\bar a)$ to indicate that $\bar a$ satisfies $\varphi(\bar x)$ in $\rela$, further 
we write $\varphi(\rela):=\{\bar a\in A^r\mid \rela\models\varphi(\bar a)\}$. Formulas without free variables are {\em sentences}.

Two structures $\rela$ and $\relb$ interpreting the same vocabulary are called {\em similar}. In this case, a {\em homomorphism} from $\rela$ to $\relb$
is a function $h: A \to B$ such that for every $R \in \sigma$
and every tuple $(a_1, \ldots, a_{\ar(R)}) \in R^{\rela}$, it holds that $(h(a_1), \ldots, h(a_{\ar(R)})) \in R^{\relb}$. We write $\rela\hto\relb$ to 
indicate that such a homomorphism exists. A {\em partial homomorphism} $h$ from $\rela$ to $\relb$ is 
either $\emptyset$ or a homomorphism from $\langle\dom(h)\rangle^{\rela}$ to $\relb$; by $\dom(h)$ we denote the domain of $h$ and by $\im(h)$ its image.

A structure $\rela$ is a \emph{core} if all homomorphisms from
$\rela$ to itself are injective.
The \emph{core of a structure} $\rela$ is a structure $\relb$
such that there is a homomorphism from $\rela$ to $\relb$, $\relb$ is a core, $B \subseteq A$, and $R^{\relb} \subseteq R^{\rela}$
for each symbol $R$.
It is well-known that each finite structure has at least one core,
and that all cores of a finite structure are isomorphic; for this reason,
one often speaks of \emph{the} core of a finite structure $\rela$,
which we denote here by $\core(\rela)$.

For example, all structures of the form $\rela^*$ are cores. Here, $\rela^*$ is the expansion 
obtained from $\rela$ by interpreting for each $a\in A$ a new unary relation symbol $C_a$ by $C^\rela_a:=\{a\}$. For a class of structures $\fancya$ we let
$$
\fancya^*:=\{\rela^*\mid\rela\in\fancya\}.
$$

\subsection{Graphs}
\label{subsect:graphs}

In this article, a \emph{graph} is a $\{E\}$-structure 
$\relg$ for a binary relation symbol $E$ such that $E^{\relg}$ is irreflexive and symmetric. A graph $\relg$ is a \emph{subgraph} of another graph $\relh$
if $G \subseteq H$ and $E^{\relg} \subseteq E^{\relh}$. 
Given a graph $\relg$ we write
$$
\refl(E^{\relg}):=E^{\relg}\cup\{(g,g)\mid g\in G \}.
$$
When $\rela$ is a $\sigma$-structure, we let $\gr(\rela)$ 
be the ``Gaifman'' graph with universe $A$ and an edge between $a,a'\in A$ if
$a\neq a'$ and there are $R\in\sigma$ and a tuple $\bar a\in R^{\rela}$ such that $a,a'$ are components of $\bar a$. 
We call $\rela$ {\em connected} if  $\gr(\rela)$ is connected. A {\em (connected) component} of $\rela$ is
a structure induced in $\rela$ on a (connected) component of $\gr(\rela)$. 

By a \emph{rooted tree}, we mean a tree that interprets
a unary relation symbol 
$\rt$ by a set containing a single element, called the \emph{root} of
the tree.
By a \emph{rooted forest}, we mean a graph $\relg$ where,
for each component $C$, the graph $\langle C \rangle^{\relg}$ is a
rooted tree.
When $a$ and $d$ are elements of a rooted tree,
we say that $a$ is an \emph{ancestor} of $d$ 
and that $d$ is a \emph{descendent} of $a$
if $a$ lies on the unique path
from the root to $d$; 
if in addition $a \neq d$, we say that
$a$ is a \emph{proper ancestor} of $d$
and that $d$ is a \emph{proper descendent} of $a$.
The \emph{height} of a rooted tree is the maximum length (number of edges)
of a  path from the root
to some vertex.
The height of a tree $\relt$ is the minimum height of all rootings of $\relt$.
The height of a forest $\relf$
is the maximum height of a (connected) component of $\relf$.

%

The \emph{tree depth} \cite{nes} of a connected graph $\relg$ is 
the minimum $h \in \nats$ such that there exists a rooted tree $\relt$ with universe $T=G$
of height $\leq h$ such that $E^\relg$ is contained in the closure of $\relt$.
The {\em closure} of $\relt$
is the set of edges $(g,g')$ such that
either $g$ is an ancestor of $g'$ in $\relt$ or vice versa.
For an arbitrary graph, its 
tree depth is defined to be the maximum tree depth taken over all components of $\relg$.

A \emph{tree decomposition} of a graph $\relg$ is a tree $\relh$ along with
an $H$-indexed family $(B_h)_{h \in H}$ of subsets of $G$ satisfying the following
conditions:
\begin{itemize}\itemsep=0pt

\item (Coverage) For every pair $(g, g') \in \refl(E^{\relg})$, there exists $h  \in H$
such that $\{ g, g' \} \subseteq B_h$.

\item 
(Connectivity) For each element $g \in G$, the set
$\{ h \mid g \in B_h \}$ is connected in $\relh$. 

\end{itemize}
The \emph{width} of a tree decomposition is $\max_{h \in H} |B_h| - 1$.
A \emph{path decomposition} of a graph is a tree decomposition
where the tree $\relh$ is a path.
The \emph{treewidth} of a graph $\relg$ is the minimum width over all 
tree decompositions of $\relg$;
likewise, the \emph{pathwidth} of $\relg$ is the minimum width over
all
path decompositions of $\relg$.
A class $\fG$ of graphs has \emph{bounded treewidth} if
there exists a constant $w$ such that each graph $\relg \in \fG$
has treewidth $\leq w$; 
the properties of \emph{bounded pathwidth} 
and \emph{bounded tree depth}
are defined similarly.

A graph $\relm$ is a \emph{minor} of a graph $\relg$
if there exists a \emph{minor map} from $\relm$ to $\relg$,
which is a map $\mu$ defined on $M$ where
\begin{itemize}\itemsep=0pt
\item for each $m \in M$,
it holds that $\mu(m)$ is a non-empty, connected subset of $G$;
\item the sets $\mu(m)$ are pairwise disjoint; and,
\item for each $(m, m') \in E^{\relm}$ there exist 
$g \in \mu(m)$ and $g' \in \mu(m')$ such that
$(g, g') \in E^{\relg}$.
\end{itemize}
When $\relg$ is a graph, we use $\minors(\relg)$ to denote the 
class of all minors of $\relg$, and we extend this notation
to a class $\fG$ setting $\minors(\fG) = \bigcup_{\relg  \in \fG} \minors(\relg)$.

The following theorem is known;
the first two parts are due to 
Robertson and Seymour's graph minor series~\cite{RobertsonV,RobertsonI}
and the third is due to 
Blumensath and Courcelle~\cite{BlumensathCourcelle10-transduction}.

\begin{theorem} 
\label{thm:excluded}
Let $\fG$ be a class of graphs.

\begin{enumerate}\itemsep=0pt

\item (Excluded grid theorem) $\fG$ has bounded treewidth
if and only if $\minors(\fG)$ does not contain all grids.

\item (Excluded tree theorem) $\fG$ has bounded pathwidth
if and only if $\minors(\fG)$ does not contain all trees.

\item (Excluded path theorem) $\fG$ has bounded tree depth
if and only if $\minors(\fG)$ does not contain all paths.

\end{enumerate}
\end{theorem}


\begin{proposition}
\label{prop:depth-gives-decomposition}
Let $\relg$ be a connected graph, and suppose that $\relt$ is a rooted tree with height $h$
witnessing that $\relg$ has tree depth $\leq h$, i.e. $\relt$  has 
height $\leq h$ and its closure contains $E^{\relg}$.
Then the tree $(T,E^{\relt})$ together with  $(B_t)_{t \in T}$ is a tree decomposition of $\relg$ of width $h$ where
$B_t = \{ a \mid a \textup{ is an ancestor of } t \}$.
\end{proposition}

To prove Proposition~\ref{prop:depth-gives-decomposition},
the key observation is the following.
For any two vertices $g, g' \in G$ connected by an edge in $\relg$,
one is an ancestor of the other in $\relt$, and if,
say, $g'$ is an ancestor of $g$, then $g, g' \in B_g$.

\section{Graph deconstructions}\label{sec:hierarchy}

In this section, 
we present the notion of \emph{graph deconstruction}
and develop its basic theory
(Section~\ref{subsect:decomp-basic});
we introduce the measure of \emph{stack depth}
(Section~\ref{subsect:stack-depth});
and,
we state and prove our theorem
presenting the graph hierarchy
(Section~\ref{subsect:hierarchy}).

\subsection{Definitions and basic properties}
\label{subsect:decomp-basic}

\begin{definition}
\label{def:deconstruction}
When $\relg$ and $\relh$ are graphs,
an \emph{$\relh$-deconstruction of $\relg$} is
an $H$-indexed family $(B_h)_{h \in H}$ 
of subsets of $G$ that satisfies the following two conditions:
\begin{enumerate}\itemsep=0pt

\item[--] (Coverage) For each pair $(g, g') \in \refl(E^{\relg})$,
there exists a pair $(h, h') \in \refl(E^{\relh})$ such that
$\{ g, g' \} \subseteq B_h \cup B_{h'}$.

\item[--] 
(Connectivity) For each element $g \in G$, the set
$\{ h \mid g \in B_h \}$ is connected in $\relh$.

\end{enumerate}
The \emph{width} of an $\relh$-deconstruction $(B_h)_{h \in H}$ 
is defined as 
$$\max_{(h, h') \in \refl(E^{\relh})} |B_h \cup B_{h'}|.$$
We will refer to the subsets $B_h$ as \emph{bags}.
\end{definition}
Note that the definition of an $\relh$-deconstruction
is similar to that of a tree decomposition, but one important
difference
is that, in the definition of $\relh$-deconstruction,
it is not required that an edge $(g, g') \in E^{\relg}$ lie inside
a single bag $B_h$, but rather, may lie inside the union $B_h \cup
B_{h'}$
of two bags where $(h, h') \in E^{\relh}$.

\begin{remark}
Another natural way to define the width of an $\relh$-deconstruction
$(B_h)_{h \in H}$
is simply as $\max_{h \in H} |B_h|$.
The theory that we develop is essentially unchanged
if one adopts this alternative definition.
\end{remark}

\begin{example}
\label{ex:grid}
Let $n \geq 1$. Let $\relh$ be the $n$-by-$n$ grid: 
it has vertices $[n]^2$ and an edge between $(i,j)$ and $(i',j')$ if 
$|i-i'|+|j-j'|=1$.
Any graph $\relg$ on $n$ vertices
has an $\relh$-decon\-struc\-tion of width $\leq
3$.
Assume without loss of generality that $G = [n]$.
The desired deconstruction is $(B_{(i,j)})_{(i,j) \in H}$ defined by
$B_{(i,j)} = \{ i, j \}$.
Coverage holds, since each pair $(i, j) \in [n]^2$ has $\{ i, j \} \subseteq B_{(i,j)}$.
Connectivity holds, since for each $i \in [n]$, 
the set $\{ h \mid i \in B_h \}$ forms a cross in the grid.
\end{example}

\begin{example}
\label{ex:self-deconstruction}
For any graph $\relg$, 
the family $(B_g)_{g \in G}$ defined by $B_g = \{ g \}$ is a
$\relg$-decon\-struc\-tion of $\relg$ of width $\le 2$.
\end{example}

\begin{proposition}
\label{prop:inheriting-connectivity}
Let $\relg,\relh$ be graphs, $S \subseteq G$  connected in $\relg$ and $(B_h)_{h\in H}$ an $\relh$-deconstruction of $\relg$.
Then $\{ h \mid S \cap B_h \neq \emptyset \}$ is connected in $\relh$.
\end{proposition}

\begin{proof}\mo{write shortened (old in file)}
 Suppose $(g, g') \in E^{\relg}$,
and define $T_g = \{ h \mid g \in B_h \}$ and
$T_{g'} = \{ h \mid  g' \in B_h \}$.
It suffices to show 
$T_g \cup T_{g'}$ is connected in $\relh$.
By coverage, there is 
$(h, h') \in \refl(E^{\relh})$
such that $\{ g, g' \} \subseteq B_h \cup B_{h'}$.
If 
$h = h'$, then 
$T_g$ and $T_{g'}$ share the vertex $h$; 
otherwise 
$(h, h') \in E^{\relh}$ with
$g \in B_h$ and $g' \in B_{h'}$, so
$h \in T_g$ and $h' \in T_{g'}$, and hence $T_g \cup T_{g'}$ is connected in $\relh$.
\end{proof}

Let $\fG$ and $\fH$ be classes of graphs.

\begin{enumerate}\itemsep=0pt

\item[--] We say that $\fG$ has $\fH$-deconstructions of width $\leq k$
if for each graph $\relg \in \fG$,
there exists a graph $\relh \in \fH$ such that
$\relg$ has an $\relh$-deconstruction of width $\leq k$.

\item[--] We say that $\fG$ has $\fH$-deconstructions of bounded width
if there exists $k \geq 1$ such that $\fG$ has $\fH$-deconstructions
of width $\leq k$.

\end{enumerate}
We will employ analogous terminology to discuss, for example,
\emph{nice deconstructions} which will be defined later.

\begin{definition}
\label{def:binary-relation}
We define the binary relation $\preceq$ on classes of graphs as follows:
$\fG \preceq \fH$
if and only if
$\fG$ has $\fH$-deconstructions of bounded width.
We write $\fG \equiv \fH$ if $\fG \preceq \fH$ and $\fH \preceq \fG$,
and we write $\fG \precneq \fH$ if $\fG \preceq \fH$ and $\fG
\not\equiv \fH$.
\end{definition}

We now present some basic properties of the relation $\preceq$.

\begin{proposition}
\label{prop:preorder}
The relation $\preceq$ is reflexive and transitive.
\end{proposition}

This implies 
that $\equiv$ is an equivalence relation. Throughout,
we will tacitly use the fact that, if $\fG \subseteq \fH$, then
$\fG \preceq \fH$.\mo{2nd sentence moved from rewritten proof (see file)}

\begin{proof}[Proof of Proposition~\ref{prop:preorder}]
Reflexivity follows from Example~\ref{ex:self-deconstruction} and transitivity from the following lemma.
%
\end{proof}

\begin{lemma}
\label{lemma:compose-deconstructions}
Let $\relg$, $\relh$, and $\reli$ be graphs;
suppose that $\relg$ has an $\relh$-deconstruction of width $\leq v$,
and that $\relh$ has an $\reli$-deconstruction of width $\leq w$.
Then, $\relg$ has an $\reli$-deconstruction of width $\leq 2vw$.
\end{lemma}

\textsc{Proof.}
Suppose that $\relg$ has the $\relh$-deconstruction
$(B_h)_{h \in H}$ of width $\leq v$,
and that $\relh$ has the $\reli$-deconstruction
$(C_i)_{i \in I}$ of width $\leq w$.
Define the family 
$(C^{+}_i)_{i \in I}$ 
by
$C^{+}_i = \bigcup_{h \in C_i} B_h$.
We claim that 
$(C^{+}_i)_{i \in I}$ 
is an $\reli$-deconstruction of $\relg$, which suffices.

To verify coverage,
let $(g, g') \in \refl(E^{\relg})$.
By coverage of $(B_h)$,
there exists $(h, h') \in \refl(E^{\relh})$ such that
$\{ g, g' \} \subseteq B_h \cup B_{h'}$.
By coverage of $(C_i)$,
there exists $(i, i') \in \refl(E^{\reli})$ such that
$\{ h, h' \} \subseteq C_i \cup C_{i'}$.
By definition of $(C^{+}_i)$, we have that
$B_h \cup B_{h'} \subseteq C_i^+ \cup C_{i'}^+$,
from which it follows that $\{ g, g' \} \subseteq C_i^+ \cup
C_{i'}^+$.

To verify connectivity, 
let $g \in G$ be an arbitrary element.
By connectivity of $(B_h)$, it holds that
$T = \{ h \mid  g \in B_h \}$ is connected in $\relh$.
It follows from
Proposition~\ref{prop:inheriting-connectivity}
that 
$U = \{ i \mid  T \cap C_i \neq \emptyset \}$ is connected in $\reli$.
We claim that $U = \{ i \mid  g \in C_i^+ \}$.
This is because 
\begin{align*}
T \cap C_i \neq \emptyset \quad
  &\Longleftrightarrow\quad \exists h \in C_i: h \in T\\
  &\Longleftrightarrow\quad \exists h \in C_i: g \in B_h\\
  &\Longleftrightarrow\quad g \in C_i^{+}.    \benda
\end{align*}

\begin{proposition}
\label{prop:equivalence-with-minors}
For any class of graphs $\fG$,
it holds that $\fG \equiv \minors(\fG)$.
\end{proposition}

\begin{proof}
It is clear that $\fG \preceq \minors(\fG)$, since 
$\fG \subseteq \minors(\fG)$.
To show that $\minors(\fG) \preceq \fG$,
we prove that when a graph $\relm$ is a minor of a graph $\relg$,
it holds that $\relm$ has a $\relg$-deconstruction of width $\leq 2$.
Let $\mu$ be a minor map from $\relm$ to $\relg$,
and define, for all $g \in G$,
the set $B_g$ to be $\{ m \mid  g \in \mu(m) \}$.
Clearly, for each $g \in G$, it holds that $|B_g| \leq 1$.
We claim that $(B_g)_{g \in G}$ is a $\relg$-deconstruction of
$\relm$.
For each $m \in M$, since $\mu(m)$ is non-empty,
there exists $g \in G$ such that $m \in B_g$.
For each $(m, m') \in E^{\relm}$, by definition of minor,
there exists $(g, g') \in E^{\relg}$ with $g \in \mu(m)$ and $g' \in
\mu(m')$;
it thus holds that $\{ m, m' \} \subseteq B_g \cup B_{g'}$.
For each $m \in M$, $\{ g \mid  m \in B_g \}$ is equal to
$\mu(m)$,
and is hence connected as $\mu$ is a minor map.
\end{proof}

We now generalize the notion of tree decomposition to arbitrary
graphs,
and then compare the resulting notion with the presented notion of
$\relh$-deconstruction.

When $\relg$ and $\relh$ are graphs, define an
\emph{$\relh$-decomposition} of $\relg$ to be an $H$-indexed family 
$(B_h)_{h \in H}$ of subsets of $G$ satisfying:
\begin{enumerate}\itemsep=0pt

\item[--] For each pair $(g, g') \in \refl(E^{\relg})$, there exists $h
  \in H$
such that $\{ g, g' \} \subseteq B_h$.

\item[--] Connectivity (as defined in Definition~\ref{def:deconstruction})

\end{enumerate}
This is a natural generalization of the definition of 
\emph{tree decomposition}:
a tree decomposition is precisely a $\relh$-decomposition
where $\relh$ is required to be a tree.

\begin{proposition}
\label{prop:decomposition-vs-deconstruction}
Let $\relg$ and $\relh$ be graphs, and let $w \geq 1$.
\begin{enumerate}\itemsep=0pt

\item $\relh$-decompositions are $\relh$-deconstructions.

\item 
If $\relh$ is a tree 
and $\relg$ has an $\relh$-deconstruction of 
width $\leq w$,
then $\relg$ has an $\relh$-decomposition of width $< 2w$.

\end{enumerate}
Consequently,
when $\fH$ is a class of trees,
it holds that $\fG \preceq \fH$ if and only if $\fG$ has
$\fH$-decompositions
of bounded width.
\end{proposition}

\begin{proof}
 (1) is straightforwardly verified.
For (2),
suppose that $(B_h)_{h \in H}$ is an $\relh$-decon\-struc\-tion of
$\relg$ where each bag has size $\leq w$.
Root the tree $\relh$, and define $p: H \to H$ as follows:
if $h$ is the root of $\relh$, define $p(h) = h$, and otherwise
define $p(h)$ to be the parent of $h$.
Define $(C_h)_{h \in H}$ by 
$C_h = B_h \cup B_{p(h)}$;
it is straightforward to verify that $(C_h)_{h \in H}$ is an
$\relh$-decomposition
of $\relg$, and each of its bags has size $\leq 2w$.
\end{proof}

\begin{remark}
We find that it is cleaner to work with the notion of $\relh$-deconstruction
than to work with the notion of $\relh$-decomposition;
this is a primary reason for our focus on the notion of
$\relh$-deconstruction.
Indeed, note that while reflexivity of the $\preceq$ relation is
straightforward
to prove, we do not know of a simple proof of reflexivity of the analogous
relation
defined via $\relh$-decomposition.
Note that while there exists a constant $w$ such that
each graph $\relg$ has a $\relg$-deconstruction of width $\leq w$
(in particular, one can take $w = 2$),
there does not exist a constant $w$ such that
each graph $\relg$ has a $\relg$-decomposition of width $\leq w$:
the bags of a $\relg$-decomposition of width $\leq w$
can cover at most a number 
of edges that is linear in the number of vertices ($|G|{w+1 \choose 2}$ many), but graphs in general may have quadratically many edges.
\end{remark}

\subsection{Stack depth}
\label{subsect:stack-depth}

Recall that $\T_d$ denotes the class of all trees of height $\leq d$.
For $h \geq 0$, $k \geq 1$, 
define $\relt_{h, k}$ to be the tree with universe $[k]^{\leq h}$
and with $E^{\relt_{h,k}} = \{ (t, ti), (ti, t) \mid t \in
[k]^{<h}, i \in [k] \}$.
Here, $[k]^{\leq h}$ and $[k]^{<h}$ 
denote the sets of strings over alphabet $[k]$
of length $\leq h$ and of length $< h$, respectively.
Define the \emph{stack depth} of a class $\fG$ of graphs to be
$$
\max \{ h \mid \forall k \geq 1, \relt_{h,k} \in \minors(\fG) \};
$$
let it be understood that this maximum is $\infty$
if the set is infinite.

\begin{proposition}\label{prop:stackdepth-pw}
A class of graphs has bounded stack depth if and only if it has bounded pathwidth.
\end{proposition}

\begin{proof}
 Let $\fG$ be a class of graphs. By Theorem~\ref{thm:excluded}, $\fG$ has unbounded pathwidth if and only if 
$\T\subseteq\minors(\fG)$. This obviously implies infinite stack depth. Conversely, infinite stack depth implies 
$\T\subseteq\minors(\fG)$:
each tree
is a subgraph of a tree of the form $\relt_{m,m}$,
and infinite stack depth gives that 
$\relt_{m,m} \in \minors(\fG)$
(for each $m \geq 1$).
\end{proof}

\begin{theorem}
\label{thm:stack-depth-decomps}
Suppose that $d, e \geq 0$ are constants and
that $\fG$ is a class of trees having stack depth $d$ and 
where each tree has height $\leq e$.
Then $\fG \preceq \T_d$.
\end{theorem}

To prove this we shall need some preparations.
Let $\relm$, $\relg$ be rooted trees.
Let us say that an $\relm$-deconstruction $(\mu(m))_{m \in M}$  of $\relg$ is
\emph{nice} if 
the following hold:
\begin{enumerate}\itemsep=0pt

\item[--] $\mu$ is a minor map from $\relm$ to $\relg$.

\item[--] $g_0 \in \mu(m_0)$, where $g_0$ and $m_0$ denote
the roots of $\relg$ and $\relm$, respectively.

\item[--] If $m'$ is a child of $m$ in $\relm$, 
$g \in \mu(m)$, $g' \in \mu(m')$, and
$(g, g') \in E^{\relg}$, then $g'$ is a child of~$g$ in $\relg$.

\end{enumerate}

\begin{lemma}
\label{lemma:transitivity-niceness}
Suppose that $\reln$, $\relm$, and $\relg$ are rooted trees,
that $(\nu(n))_{n \in N}$ is a nice $\reln$-deconstruction of $\relm$,
and
that $(\mu(m))_{m \in M}$ is a nice $\relm$-deconstruction of $\relg$.
Then $(\mu(\nu(n)))_{n \in N}$ is a nice $\reln$-deconstruction of
$\relg$,
where here $\mu(\nu(n))$ denotes $\bigcup_{m \in \nu(n)} \mu(m)$.
\end{lemma}

\begin{proof}
 It is straightforward to verify that $n \mapsto \mu(\nu(n))$ is a
minor map
from~$\reln$ to~$\relg$.
It follows from the proof of Lemma~\ref{lemma:compose-deconstructions}
that $(\mu(\nu(n)))_{n \in N}$ is an $\reln$-deconstruction of $\relg$.
We verify that this deconstruction is nice, as follows.
By the niceness of $(\mu(m))_{m \in M}$ and $(\nu(n))_{n \in N}$,
we have that $g_0 \in \mu(m_0)$ and $m_0 \in \nu(n_0)$,
so $g_0 \in \mu(\nu(n_0))$
(here,  $g_0$, $m_0, n_0$ denote the roots of
$\relg$, $\relm$, $\reln$, respectively).
Next, suppose that $n'$ is a child of $n$ in $\reln$,
that $g \in \mu(\nu(n))$, that $g' \in \mu(\nu(n'))$,
and that $(g, g') \in E^{\relg}$.
There are $m, m' \in M$
such that
$g \in \mu(m)$, $m \in \nu(n)$,
$g' \in \mu(m')$, and $m' \in \nu(n')$.
Since $n' \neq n$, we have $m \neq m'$; 
as $(g, g') \in E^{\relg}$, we then have $(m, m') \in E^{\relm}$, 
and
from the niceness of $(\nu(n))_{n \in N}$, we have that $m'$ is a child of $m$ 
in $\relm$.
By the niceness of $(\mu(m))_{m \in M}$ then $g'$ is a child
of~$g$.
\end{proof}

For $d \geq 0$ and $k \geq 1$, let us say that a node $u$ of a rooted tree 
has property $P(d,k)$ if either $d = 0$ or
 $d > 0$ and $u$ has $k$ pairwise incomparable descendents
each having property $P(d-1, k)$.
(Here, we consider two nodes $v, v'$ to be incomparable if neither is
an ancestor of the other.)
Let us say that a rooted tree has property $P(d, k)$ if its root has
property $P(d,k)$.
Observe that 
such a tree 
contains $\relt_{d,k}$ as a minor.

\begin{lemma}
\label{lemma:property-transfer}
Suppose that $\relm$ and $\relg$ are rooted trees and that
$(\mu(m))_{m \in M}$ is a nice $\relm$-deconstruction of $\relg$.
For $d \geq 0$ and $k \geq 1$,
if $\relm$ has property $P(d,k)$, then so does $\relg$.
\end{lemma}

\begin{proof}
 For each $m \in M$, since $\mu(m)$ is connected,
it has a unique highest element, by which we mean the element with
shortest
distance to the root; denote this element by $\hi(\mu(m))$.
Suppose that
 $m'$ is a child of $m$ in $\relm$;
the parent of $\hi(\mu(m'))$ must, by the definition of
deconstruction,
lie in $\mu(m'')$ where $m''$ is adjacent to $m'$; but it must be that
$m'' = m$ by the niceness of $(\mu(m))_{m \in M}$.
Thus, if $m'$ is a child of $m$ in $\relm$,
then $\hi(\mu(m'))$ is a descendent of $\hi(\mu(m))$.
It follows by induction on $d$ that,
if a node $m \in M$ has property $P(d,k)$ in $\relm$, then
$\hi(\mu(m))$ has property $P(d,k)$ in $\relg$.
\end{proof}



\begin{proof}[Proof of Theorem~\ref{thm:stack-depth-decomps}.]
Let $K \geq 1$ be a constant.
Suppose that $\fG$ is a class of rooted trees
of height $\leq e$
which
do not have property $P(d+1, K)$;
we prove that
$\fG$ has nice $\T_d$-deconstructions of bounded width.
(This suffices, since the assumption that $\fG$ has stack depth $d$
implies that there is a constant $K \geq 1$ such that 
$\relt_{d+1, K} \notin \minors(\fG)$, which in turn implies that
the trees in $\fG$ do not have property $P(d+1, K)$.)

We proceed by induction on $d$.

\medskip

{\em Case $d = 0$: }  We have that
the trees in $\fG$ do not have property $P(1, K)$.
Consider a rooted tree from $\fG$.
The number of leaves is bounded above by $K$;
since each node is the ancestor of a leaf, the total number of nodes
is bounded above by $K (e + 1)$.
Thus $\fG$ has nice $\T_0$-deconstructions of width $\leq K(e + 1)$.

\medskip

{\em Case $d > 0$: }  
We argue by induction on $e$.
If $e \leq d$, then we have that $\fG \subseteq \T_d$ and we are done
(for each $\relg \in \fG$, 
use the $\relg$-deconstruction of $\relg$
discussed in conjunction with
reflexivity in Proposition~\ref{prop:preorder}).
So suppose that $e > d$.

Define a class of trees $\fG'$ as follows:
for each tree $\relg$ in~$\fG$, and for each child $c$ of the root of
$\relg$,
if $c$ does not have property $P(d, K)$, then place the subtree of
$\relg$
rooted at $c$ in $\fG'$.
The trees in $\fG'$ have bounded height and do not have property 
$P(d, K)$;
so, by induction,
there is a constant $w$ such that
$\fG'$ has nice $\T_{d-1}$-deconstructions of width $\leq w$.

Let $\relg$ be a tree in $\fG$.
Let $b^1, \ldots, b^L$ denote the children of the root $g_0$ that have
property
$P(d, K)$, and let $c^1, \ldots, c^Q$ denote the remaining children of
the root.
Since the root of $\relg$ does not have property $P(d + 1, K)$,
we have that $L < K$.
Let $\relg^1, \ldots, \relg^Q$ denote the subtrees of $\relg$
rooted at $c^1, \ldots, c^Q$, respectively.
Each $\relg^i$ is in $\fG'$, so 
for each $i \in [Q]$, there is a tree $\relt^i \in \T_{d-1}$
such that $\relg^i$ has a nice $\relt^i$-deconstruction
$(\mu^i(t))_{t \in T^i}$ of width $\leq w$.

Now define the tree $\relh$ to be the minor of $\relg$ obtained from
$\relg$
by contracting together the vertices $\{ g_0, b^1, \ldots, b^L \}$
to obtain $h_0$, 
and by replacing each $\relg^i$ with $\relt^i$.
Observe that the height of $\relh$ is $\leq e-1$.
The following map $\mu$ is a minor map from $\relh$ to $\relg$:
$\mu(h_0) = \{ g_0, b^1, \ldots, b^L \}$,
$\mu(t)$ is equal to $\mu^i(t)$ if $t \in T^i$,
and $\mu(h) = \{ h \}$ for all other vertices $h \in H$.
It is straightforward to verify that $(\mu(h))_{h \in H}$ gives a nice
$\relh$-deconstruction of $\relg$ having width $\leq \max(K, w)$.

Let $\fH$ denote the class of all trees $\relh$ obtained from $\relg \in
\fG$
in this way. 
We just saw that $\fG$ has nice $\fH$-deconstructions of bounded width.
Since $\fH$ has height $\leq e-1$ and does not have property 
$P(d+1, K)$
by Lemma~\ref{lemma:property-transfer}, 
by induction, $\fH$ has nice $\T_{d}$-deconstructions of bounded
width.
As a consequence of Lemma~\ref{lemma:transitivity-niceness},
we obtain that $\fG$ has nice $\T_d$-deconstructions of bounded width.
\end{proof}

\subsection{Hierarchy}
\label{subsect:hierarchy}

Recall that $\L,\T,\P$ denote the classes of graphs, trees and paths respectively, and $\T_d,\F_d$ denote 
the classes of trees respectively forests of height at most $d$.

\begin{theorem}[Graph hierarchy theorem]
\label{thm:graph-hierarchy}
The hierarchy
\begin{equation}\label{eq:hierarchy}\tag{$*$}
\T_0 \precneq \F_0 \precneq \T_1 \precneq \F_1 \precneq \cdots
\precneq \P \precneq \T \precneq \L
\end{equation}
presents correct relationships, and 
is comprehensive in that each class of graphs is equivalent 
(under $\equiv$) to one of the classes therein.
\end{theorem}

We break the proof into several lemmas.
We begin by observing that the established conditions of
bounded treewidth, bounded pathwidth, and bounded tree depth
can be formulated using the $\preceq$ relation and the defined graph classes.

\begin{proposition}
\label{prop:bounded-characterizations}
Let $\fG$ be a class of graphs.
\begin{enumerate}\itemsep=0pt

\item $\fG$ has bounded treewidth if and only if $\fG \preceq \T$.

\item $\fG$ has bounded pathwidth if and only if $\fG \preceq \P$.

\item $\fG$ has bounded tree depth if and only if 
there exists $d \geq 0$
such that $\fG \preceq \F_d$.

\end{enumerate}
\end{proposition}

\begin{proof}
 The first two claims follow immediately from 
Proposition~\ref{prop:decomposition-vs-deconstruction}.
For the third claim, we reason as follows.
For the forward direction,
let $d$ be an upper bound on the tree depth of $\fG$,
and
let $\relg$ be a graph in $\fG$;
Proposition~\ref{prop:depth-gives-decomposition}
(along with 
Proposition~\ref{prop:decomposition-vs-deconstruction})
implies that each component of $\relg$ has a $\T_d$-deconstruction
of width $\leq d$.
For the backward direction, 
Proposition~\ref{prop:decomposition-vs-deconstruction}
gives a constant $w$ such that
$\fG$ has $\F_d$-decompositions of width $< w$,
which implies that $\fG$ has tree depth $\leq wd$
(see~\cite[Remark 4.3(a)]{BlumensathCourcelle10-transduction}).
\end{proof}

In the next two lemmas, we present the negative results needed 
to give the hierarchy, showing that various pairs of graph classes
are not related by $\preceq$.

\begin{lemma}
\label{lemma:placing-G}
Let $\fG$ be a class of graphs.
\begin{enumerate}\itemsep=0pt

\item If $\L \not\preceq \fG$, then $\fG \preceq \T$.

\item If $\T \not\preceq \fG$, then $\fG \preceq \P$.

\item If $\P \not\preceq \fG$, then $\fG$ has bounded tree depth.

\end{enumerate}
\end{lemma}

\begin{proof}\mo{shorter phrasing}
To show (3), assume $\P \not\preceq \fG$.
By Proposition~\ref{prop:equivalence-with-minors},
$\P \not\preceq \minors(\fG)$,
so $\P \not\subseteq \minors(\fG)$,
and hence
$\fG$ has bounded tree depth by 
Theorem~\ref{thm:excluded}.

For (2) we reason analogously: 
if $\T \not\preceq \fG$, 
then $\fG$ has bounded pathwidth
%
by Proposition~\ref{prop:equivalence-with-minors}
and
Theorem~\ref{thm:excluded}
and hence $\fG \preceq \P$
by Proposition~\ref{prop:bounded-characterizations}.

Also (1) is proved analoguously, but
in order to use Theorem~\ref{thm:excluded}, 
we need to prove that $\L \preceq \R$, where
$\R$ denotes the class of all grids;
then $\L \not\preceq \fG$ implies $\R \not\preceq \fG$.
That $\L \preceq \R$ follows from Example~\ref{ex:grid}:
each graph has $\R$-deconstructions of width $\leq 2$.
\end{proof}

\begin{lemma}
\label{lemma:top-hierarchy}
The following relationships hold.
\begin{enumerate}\itemsep=0pt

\item $\L \not\preceq \T$.

\item $\T \not\preceq \P$.

\item $\P \not\preceq \F_d$, for all $d \geq 0$.

\end{enumerate}
\end{lemma}

\begin{proof}
Immediate from
Theorem~\ref{thm:excluded}
and Proposition~\ref{prop:bounded-characterizations}.
\end{proof}

\begin{lemma}
\label{lemma:bounded-td-hierarchy}
For each $d \geq 0$, the following hold.
\begin{enumerate}\itemsep=0pt

\item $\F_d \not\preceq \T_d$.
\item $\T_{d+1} \not\preceq \F_d$.

\end{enumerate}
\end{lemma}

\begin{proof}
We prove this by induction on $d$.
When $d = 0$, the claim that $\F_d \not\preceq \T_d$ is
clear, since $\T_0$ contains only one-vertex trees,
but $\F_0$ contains graphs with arbitrarily many vertices.
The claim that $\T_{d+1} \not\preceq \F_d$
always follows from $\F_d \not\preceq \T_d$, as follows:
the assumption $\T_{d+1} \preceq \F_d$ 
implies 
$\T_{d+1} \preceq \T_d$
by Proposition~\ref{prop:inheriting-connectivity},
and this 
implies 
$\F_d \preceq \T_d$.

Let $d > 0$; we need to prove that $\F_d \not\preceq \T_d$.
Suppose for a contradiction that $w$ is a constant
such that $\F_d$ has $\T_d$-deconstructions of width $\leq w$.
We show that $\T_d \preceq \F_{d-1}$, which contradicts the induction
hypothesis.
Given an arbitrary tree $\relt$ in $\T_d$,
create~$w+1$ copies of it; the resulting graph has
 an~$\relh$-deconstruction $(B_h)_{h \in H}$
of width $\leq w$, with $\relh \in \T_d$.
Root~$\relh$ with a vertex $h_0$ to witness that its height is $\leq
d$.
Since the bag $B_{h_0}$ has size $\leq w$, there must be a copy of~$\relt$
that is disjoint from $B_{h_0}$; call this the \emph{key copy}
and denote its universe by $K$.
By Proposition~\ref{prop:inheriting-connectivity},
it is possible to take a subtree $\relh'$ of~$\relh$ that excludes
$h_0$
such that $(B_h \cap K)_{h \in H'}$ gives an $\relh'$-deconstruction of
the key copy.  The graph $\relh'$ has height $\leq d-1$.
We thus showed that $\T_d$ has $\T_{d-1}$-deconstructions of width
$\leq w$, implying that $\T_d \preceq \F_{d-1}$, as desired.
\end{proof}

%

\begin{lemma}
\label{lemma:graphs-to-forests}
Let $\fG$ be a class of graphs.
If $\fG$ has bounded tree depth, then there exists a class $\fH$
of forests of bounded height such that $\fG \equiv \fH$ and
such that $\fG$ and $\fH$ have the same stack depth.
(Namely, one can take $\fH$ to be the class of all forests in $\minors(\fG)$.)
\end{lemma}

\begin{proof}
 Let $\fH$ be as described.  The class $\fH$ has bounded height,
as $\fG$ has bounded tree depth.  The classes $\fG$ and $\fH$
have the same tree minors, so they have the same stack depth.
We have $\fH \preceq \fG$ by
Proposition~\ref{prop:equivalence-with-minors}.

The proof of $\fG \preceq \fH$ is 
along the lines of that 
of~\cite[Lemma 4.8]{BlumensathCourcelle10-transduction}.
We provide it here for completeness.
For each graph $\relg \in \fG$,
let $\relh$ be a forest that contains a depth-first search tree of
each component of $\relg$
(for precise details, we refer to Diestel~\cite{Diestel-GraphTheory4thEdition}, 
who speaks of \emph{normal spanning trees}; see Proposition 1.5.6).
Since $\relh$ is a subgraph of $\relg$, we have $\relh \in \fH$.
Each component of $\relg$ is contained
in the closure of its corresponding component in $\relh$
(in the sense of the definition of tree depth), so 
Proposition~\ref{prop:depth-gives-decomposition} and the fact that
$\fH$ has bounded height
imply that $\fG \preceq \fH$.
\end{proof}

\begin{lemma}
\label{lemma:stack-depth-in-hierarchy}
Let $\fG$ be a class of forests having bounded height,
and let $d$ denote the stack depth of $\fG$.
It holds either that $\fG \equiv \T_d$ or that $\fG \equiv \F_d$.
\end{lemma}

\begin{proof} 
By Lemma~\ref{lemma:bounded-td-hierarchy}, not both $\fG \equiv \T_d$ and $\fG \equiv \F_d$ can hold.
Let $\fC$ be the class of connected graphs that appear as components of
graphs in $\fG$; note that $d$ is the stack depth of~$\fC$. 
By Theorem~\ref{thm:stack-depth-decomps} we have $\fC\preceq \T_d$ and hence $\fG \preceq \F_d$

For each $k \geq 1$ and each graph $\relg \in \fG$,
define $\relg(k)$ to be the number of components of $\relg$
having the tree $\relt_{d,k}$ as a minor.
We consider two cases.

\medskip

{\em Case 1: } 
Suppose that, for all $k \geq 1$, the set $\{ \relg(k) \mid  \relg \in \fG \}$ has infinite size.
We claim that $\fG \equiv \F_d$. 
We have to show that $\F_d \preceq\fG$.
For each $k \geq 1$, let us use $k \times \relt_{d,k}$
to denote the graph consisting of $k$ disjoint copies of $\relt_{d,k}$.
Then
$\{ k \times \relt_{d,k} \mid  k \geq 1 \} \subseteq \minors(\fG)$ by assumption.
Each graph in $\F_d$ is isomorphic to a subgraph of a graph of the form $k \times
\relt_{d,k}$.
Hence,
by Proposition~\ref{prop:equivalence-with-minors},
it holds that $\F_d \preceq \minors(\fG) \equiv \fG$.

\medskip

{\em Case 2: } 
When the assumption of the first case does not hold,
one can choose a sufficiently large $K \geq 1$
such that, for all $\relg \in \fG$, it holds that $\relg(K) \leq K$.
We claim that $\fG \equiv \T_d$.
That $\T_d \preceq \fG$ follows from the hypothesis that $d$ is the
stack depth of $\fC$. We show that $\fG \preceq \T_d$.
Let $\fH$ be the subset of $\fC$ that contains a graph $\relh \in \fC$
if and only if $\relt_{d,K}$ is not a minor of $\relh$.
By Theorem~\ref{thm:stack-depth-decomps},
it holds that $\fH \preceq \T_{d-1}$; let $w \geq 1$ be such that
$\fH$ has $\T_{d-1}$-deconstructions of width $\leq w$.
Let $\relg \in \fG$; let $\relg_1, \ldots, \relg_L$ be the components
of $\relg$ having $\relt_{d,K}$ as a minor,
and let $\relh_1, \ldots, \relh_M$ be the other components of $\relg$.
By the choice of $K$, it holds that $L \leq K$.
Since $\fC \preceq\T_d$, there exists $v \geq 1$ such that
$\fC$ has $\T_d$-deconstructions of width $\leq v$.
Let $\relt \in \T_d$ be a sufficiently large tree so that
each 
$\relg_i$
has a $\relt$-deconstruction 
$(B^i_{t})_{t \in T}$ of
width $\leq v$;
then, the disjoint union of the $\relg_i$ has a $\relt$-deconstruction
of width $\leq vL$,
namely, $(B^1_t \cup \cdots \cup B^L_t)_{t \in T}$.
Each $\relh_j$ has a $\relt_j$-deconstruction of width $\leq w$, 
where $\relt_j \in \T_{d-1}$.
Let $\relt'$ be equal to $\relt$ but augmented so that
the root of each $\relt_j$ is a child of the root of $\relt$;
we have that the height of $\relt'$ is $d$.
The graph $\relg$ has a $\relt'$-deconstruction where each bag is
defined
as it was in the respective $\relt$-deconstruction or $\relt_j$-deconstruction;
this $\relt'$-deconstruction has width $\leq \max(vL, w)$.
\end{proof}

The following is a consequence of the previous two lemmas.

\begin{lemma}
\label{lemma:stack-depth-placement}
Let $\fG$ be a class of graphs having bounded tree depth, and let $d
\geq 0$.
The class $\fG$ has stack depth $d$ if and only if
$\fG \equiv \T_d$ or $\fG \equiv \F_d$.
\end{lemma}

\begin{proof}
 We first prove the forward direction.  
Suppose that $\fG$ has stack depth $d$.
By Lemma~\ref{lemma:graphs-to-forests}, 
there exists a class of forests $\fH$ having bounded height and stack depth~$d$
such that $\fG \equiv \fH$.
By Lemma~\ref{lemma:stack-depth-in-hierarchy},
$\fH \equiv \T_d$ or $\fH \equiv \F_d$,
implying 
$\fG \equiv \T_d$ or $\fG \equiv \F_d$.

Conversely, suppose  $\fG$ does not have stack depth $d$. 
By bounded tree depth and Proposition~\ref{prop:stackdepth-pw}, $\fG$
has stack depth $d'\in\mathbb N$ for some $d'\neq d$.
By the forward direction $\fG\equiv\T_{d'}$ or $\fG\equiv\F_{d'}$. In both cases, 
$\fG\not\equiv\T_{d}$ and $\fG\not\equiv\F_{d}$ by Lemma~\ref{lemma:bounded-td-hierarchy}.
\end{proof}

\begin{proof}[Proof of Theorem~\ref{thm:graph-hierarchy}.]
It is clear that
$$\T_0 \preceq \F_0 \preceq \T_1 \preceq \F_1 \preceq \cdots
\preceq \P \preceq \T \preceq \L.$$
Lemmas~\ref{lemma:top-hierarchy}
and~\ref{lemma:bounded-td-hierarchy} imply that
none of the displayed $\preceq$ can be reversed.
To prove that the hierarchy is comprehensive, let $\fG$
be an arbitrary class of graphs.  
If $\L \preceq \fG$, then clearly $\fG \equiv \L$ and we are done.
Otherwise $\fG \preceq \T$ by Lemma~\ref{lemma:placing-G}~(1).
If $\T \preceq \fG$, then $\fG \equiv \T$ and we are done.
Otherwise, $\fG \preceq \P$ by Lemma~\ref{lemma:placing-G}~(2).
If $\P \preceq \fG$, then $\fG \equiv \P$ and we are done.
Otherwise $\fG$ has bounded tree depth  by Lemma~\ref{lemma:placing-G}~(3).
By Proposition~\ref{prop:stackdepth-pw} there is $d\in\mathbb N$ such that $\fG$ has stack depth $d$.
Then $\fG \equiv \T_d$ or $\fG \equiv \F_d$ by Lemma~\ref{lemma:stack-depth-placement}.
\end{proof}

\section{Grohe's theorem} \label{sect:grohes-theorem}

In this section, we use the notion of graph deconstruction
to give a novel proof of Grohe's theorem, which 
establishes the hardness of the homomorphism problem
on any class of structures whose cores have unbounded treewidth.
We believe that our proof constitutes a 
 modular, relatively transparent, and relatively simple
alternative to the original
proof~\cite{Grohe07-otherside}.
Other than the definition of graph deconstruction,
the only element needed from the previous section
is the fact that, if a graph class $\fG$
has unbounded treewidth, then $\L \preceq \fG$
(this follows from
Proposition~\ref{prop:bounded-characterizations}
and Lemma~\ref{lemma:placing-G}.)

For the sake of brevity and because it is unnecessary for 
our purposes here, we do not introduce here a full framework
for parameterized complexity.
(We do carry this out in the next section, 
where we in particular introduce our notion 
of quantifier-free reduction.)
We introduce the following definitions to be used
in the scope of this section.
For each class $\fancya$ of structures,
define $\HOMP{\fancya}$
to be the problem whose instances 
are pairs $(\rela, \relb)$ of similar structures
where $\rela \in \fancya$,
and the question is to decide
whether or not $\rela \hto \relb$.
We consider such a problem $\HOMP{\fancya}$ to be
\emph{tractable} if there exists
a computable function $f$ and a polynomial-time algorithm
$g$ such that,
on each instance $(\rela, \relb)$ of $\HOMP{\fancya}$,
$g$ run on input $(f(\rela),\relb)$ decides if $\rela \hto \relb$.
It can be recognized that this definition is
equivalent to that of fixed-parameter tractability,
where the structure $\rela$ is taken to be the parameter;
see
the characterization of fixed-parameter
tractability in terms of precomputation on the parameter~\cite[Theorem 1.37]{FlumGrohe06-parameterizedcomplexity}. 
It is well-known that the tractability
of $\HOMP{\L^*}$ is equivalent to the complexity class collapse W[1] = FPT.

We prove the following formulation of 
Grohe's theorem~\cite{Grohe07-otherside}.

\begin{theorem}
\label{thm:grohe}
Assume that  $\fancya$ is a computably enumerable 
class of structures having bounded arity.
If the graphs of the cores of $\fancya$
have unbounded treewidth,
then the problem $\HOMP{\fancya}$ is not tractable,
unless $\HOMP{\L^*}$ is as well.
\end{theorem}

%

At the heart of our proof are three polynomial-time reductions,
presented in the following three lemmas.
In each case, 
we describe the output of the claimed polynomial-time algorithm;
it is readily verified that the output can be produced
in polynomial time.
The second and third lemmas are based on results that
appeared in previous work~\cite{ChenMueller13-fineclass-arxiv}.

\begin{lemma}
\label{lemma:polytime-deconstruction}
For each $k \geq 1$,
there exists a polynomial-time algorithm that,
given graphs $\relg$ and $\relh$,
a $\relh$-deconstruction $(B_h)_{h \in H}$ of $\relg$
of width $\leq k$,
and a structure $\reld$ similar to~$\relg^*$,
outputs a structure $\reld'$ such that
$\relg^* \hto \reld$ iff $\relh^* \hto \reld'$
\end{lemma}

\begin{proof}
The structure $\reld'$ is defined as follows.
Its universe $D'$ is the set of all partial homomorphisms
$f$ from $\relg^*$ to $\reld$ with $|\dom(f)| \leq k$.
The relation $E^{\reld'}$ is defined as the set of pairs
$(f,f') \in D' \times D'$
such that  $f \cup f'$ 
(as a set of ordered pairs)
is a partial homomorphism from $\relg^*$ to $\reld$.
Each relation $C_h^{\reld'}$ is defined as 
$\{ f \in D' ~|~ \dom(f) = B_h \}$.

Suppose that $e$ is a homomorphism from $\relg^*$ to $\reld$.
Then the mapping $e': H \to D'$ defined by
$e'(h) = e \res B_h$ (the restriction of $e$ to $B_h$) 
is  
a homomorphism from
$\relh^*$ to $\reld'$.

Suppose that $e'$ is a homomorphism from $\relh^*$ to $\reld'$.
We define a map $e: G \to D$ as follows.
For each $g \in G$, by the connectivity condition
(Definition~\ref{def:deconstruction})
and the definition of $E^{\reld'}$,
all maps of the form $e'(h)$ with $g \in \dom(e'(h))$
send $g$ to the same value.
Define $e(g)$ to be that value; we have that 
$e'(h) \subseteq e$.
To verify that $e$ is a homomorphism from $\relg^*$ to $\reld$,
since each relation of $\relg^*$ has arity $1$ 
or is the relation $E^{\relg}$,
it suffices to argue that for any pair $(g, g') \in \refl(E^{\relg})$,
the map $e \res \{ g, g' \}$ is a partial homomorphism
from $\relg^*$ to $\reld$.
Let $(g, g')$ be such a pair;
by the coverage condition
(Definition~\ref{def:deconstruction})
there exists $(h, h') \in \refl(E^{\relh})$
such that $\{ g, g' \} \subseteq B_h \cup B_{h'}$.
We have that $e'(h) \cup e'(h')$
is a partial homomorphism from $\relg$ to $\reld$
(this is clear if $h = h'$; if $h \neq h'$,
this follows from from the definition of $E^{\reld'}$).
This concludes the proof, as 
$e'(h) \cup e'(h') \subseteq e$.
\end{proof}

\begin{lemma}
\label{lemma:polytime-undo-graph}
For each $r \geq 1$,
there exists a polynomial-time algorithm that,
given a structure $\rela$ whose relations have arity $\leq r$
and a structure $\reld$ similar to $\gr(\rela)^*$,
outputs a structure $\reld'$ such that
$\gr(\rela)^* \hto \reld$ iff $\rela^* \hto \reld'$.
\end{lemma}

\begin{proof}
The structure $\reld'$ has universe
$D' = A \times D$.
Each relation $C_a^{\reld'}$ is defined as
$\{ a \} \times C_a^{\reld}$,
and for each relation symbol $R$ of the structure of $\rela$,
define
$R^{\reld'}$ to be the set of all $k$-tuples
$((a_1, d_1), \ldots, (a_k, d_k))$ on $D'$
such that $k = ar(R)$ and for all $i, j \in [k]$,
it holds that $(a_i, a_j) \in E^{\gr(\rela)^*}$ implies
$(d_i, d_j) \in E^{\reld}$.

Suppose that $e$ is a homomorphism from
$\gr(\rela^*)$ to $\reld$.
Then it is straightforward to verify that $e': A \to D'$
defined by $e'(a) = (a, e(a))$
is a homomorphism from $\rela^*$ to $\reld'$.

Suppose that $e'$ is a homomorphism from
$\rela^*$ to $\reld'$.
By the definition of the relations $C_a^{\reld'}$,
each element $a$ is mapped by $e$ to an element
of the form $(a, d)$ with $d \in C_a^{\reld}$.
Define $e: A \to D$ so that, for each $a \in A$,
it holds that $e'(a) = (a, e(a))$.
We have that $e$ is a homomorphism from 
$\gr(\rela)^*$ to $\reld$:
when $(a, a') \in E^{\gr(\rela)^*}$,
there exists a tuple $(a_1, \ldots, a_k)$
in a relation $R^{\rela^*}$ of $\rela^*$
with $a$ and $a'$ among its entries,
so $(e(a), e(a')) \in E^{\reld}$
follows from the fact that $e'$ is a homomorphism
and the definition of $R^{\reld'}$.
\end{proof}

\begin{lemma}
\label{lemma:polytime-uncore}
There exists a polynomial-time algorithm that,
given a core $\rela$ and a structure $\reld$ similar to $\rela^*$,
outputs a structure $\reld'$ such that
$\rela^* \hto \reld$ iff $\rela \hto \reld'$.
\end{lemma}

\begin{proof}
Define $\reld'$ as the structure 
$\langle \{ (a,d) \in A \times D ~|~ d \in C_a^{\reld} \}
\rangle^{\rela \times \reld_{\sigma}}$,
where $\reld_{\sigma}$ denotes the restriction of $\reld$
to the vocabulary $\sigma$ of $\rela$.
(If the specified set of pairs is empty, 
the algorithm outputs a fixed \emph{no} instance.)

Suppose that $e$ is a homomorphism from $\rela^*$ to $\reld$;
then, the map $e': A \to D'$ defined by
$e'(a) = (a,e(a))$ is straightforwardly verified
to be a homomorphism from $\rela$ to $\reld'$.

Suppose that $e'$ is a homomorphism from $\rela$ to $\reld'$.
Here, for any homomorphism $g'$ from $\rela$ to $\reld'$,
we let $g'_1: A \to A$ and $g'_2: A \to D$ denote
the maps such that $g'(a) = (g'_1(a),g'_2(a))$ for each $a \in A$.
We have that $e'_1$ is a homomorphism from $\rela$ to itself;
since $\rela$ is a core, $e'_1$ is a bijection.
It follows that each finite power of $e'_1$, and hence
$e'^{-1}_1$, is a homomorphism from $\rela$ to itself.
Composing $e'^{-1}_1$ with $e'$,
we obtain a homomorphism $f'$ from $\rela$ to $\reld'$
such that $f'_1$ is the identity map on $A$.
By definition of $\reld'$, for each $a \in A$
it holds that $f'_2(a) \in C_a^{\reld}$,
and for each $R \in \sigma$ 
and each tuple $(a_1, \ldots, a_k) \in R^{\rela}$,
it holds that $(f_2(a_1), \ldots, f_2(a_k)) \in R^{\reld_{\sigma}}$
by definition of $\reld'$.
Hence $f'_2$ is a homomorphism from $\rela^*$ to $\reld$.
\end{proof}

\begin{proof}[Proof of Theorem~\ref{thm:grohe}.]
Assume that the problem $\HOMP{\fancya}$ is tractable
via $(f',g')$;
we show that the problem $\HOMP{\L^*}$ is tractable.
Let $r \geq 1$ be a bound on the arity of $\fancya$.
As noted at the beginning of the section,
our hypothesis on $\fancya$ implies that
there exists
$k \geq 1$ such that, for each graph $\relg$,
there exists a structure $\rela \in \fancya$ 
that \emph{corresponds} to $\relg$, by which we mean that
the
core $\relc$ of $\rela$ has a graph $\relh = \gr(\relc)$
such that $\relg$ has a $\relh$-deconstruction $(B_h)_{h \in H}$
of width $\leq k$.

The following pair $(f,g)$ establishes the tractability 
of $\HOMP{\L^*}$.
Given an instance $(\relg^*, \relb)$ thereof, 
the algorithm first performs a computation depending only on
$\relg^*$.  In particular, it enumerates the structures
in $\fancya$ until it finds a structure $\rela \in \fancya$
that corresponds to $\relg$; it outputs the core $\relc$ of $\rela$,
the core's graph $\relh = \gr(\relc)$, 
and the $\relh$-deconstruction $(B_h)_{h \in H}$ of $\relg$
having width $\leq k$.
All of this information plus the value of $f'(\rela)$
is the output of $f(\relg^*)$.
Then,   $g$ is defined
to be the polynomial-time algorithm
that
invokes the algorithm of Lemma~\ref{lemma:polytime-deconstruction}
on $(\relg^*, \reld)$
to obtain an instance $(\relh^*, \reld')$;
invokes the algorithm of Lemma~\ref{lemma:polytime-undo-graph}
on $(\relh^*, \reld')$
to obtain an instance $(\relc^*, \reld'')$;
and, then
invokes the algorithm of Lemma~\ref{lemma:polytime-uncore}
on $(\relc^*, \reld'')$
to obtain an instance $(\relc, \reld''')$.
We have $\relc \hto \reld'''$ if and only if
$\rela \hto \reld'''$.
Thus, the algorithm invokes the algorithm $g'$ 
on $(f'(\rela), \reld''')$ and outputs the answer of this invocation.
\end{proof}

\section{Complexity classification}\label{sec:classification}

In this section we study the complexity of the parameterized homomorphism problems associated to classes of structures $\fancya$: 
$$
\HOMP{\fancya}:=\big\{(\relb,\numb{\rela})\mid \rela\in\fancya\ \&\  \rela\hto\relb\big\}.
$$
Here, $\numb{\rela}$ is a natural number coding the structure $\rela$ in some natural way.
The goal of this section is to 
show that the complexities of homomorphism problems are captured in a strong sense by the hierarchy 
from Section~\ref{sec:hierarchy}, namely with respect  to a computationally very weak notion of reduction which we call quantifier-free after a pre-computation ({\em qfap}).

We recall some basic notions from parameterized complexity theory in the next subsection; define qfap-reductions in Section~\ref{sec:qfap};
and consider the homomorphism problem for graph classes 
 and 
subsequently 
for general classes of structures in 
Sections~\ref{sec:graphclass}
and~\ref{sec:strclass}.

\subsection{Parameterized and descriptive complexity}

A {\em parameterized problem $Q$} is a subset of $\{0,1\}^*\times\nats$. By a {\em classical problem} we mean a
subset of $\{0,1\}^*$.
Given an instance $(x,k)$ of $Q$ we refer to $k$ as its {\em parameter}. The {\em $k$th slice} of $Q$ is the classical problem
$\{x\in\{0,1\}^*\mid (x,k)\in Q\}$.

Following~\cite{FlumGrohe03-describing},
we say $Q$ is in L {\em after a pre-computation} if there is a computable function
$a:\nats\to\{0,1\}^*$ and a classical problem $P\subseteq \{0,1\}^*$ in L such 
that for all $(x,k)\in\{0,1\}^*\times\nats$
$$
(x,k)\in Q\Longleftrightarrow \langle x,a(k)\rangle\in P,
$$
where $\langle\cdot,\cdot\rangle$ is some standard pairing function for binary strings.
Equivalently, this means that $(x,k)\stackrel{?}{\in}Q$ is decidable in space $O(f(k)+\log n)$ for some computable $f:\nats\to\nats$.
The class of such problems is denoted {\em para-L}.
This mode of speech makes sense not only for L but for any classical complexity class, and we refer to \cite{FlumGrohe03-describing}
for the corresponding theory. For example, FPT is the class of parameterized problems which are in P after a pre-computation.
A {\em parameterized reduction} from a parameterized problem $Q$ to another $Q'$ is a function $r:\{0,1\}^*\times\nats\to\{0,1\}^*\times\nats$ 
such that there is a computable $g:\nats\to\nats$ such that for all $(x,k)\in\{0,1\}^*\times\nats$ we have for $(x',k'):=r((x,k))$ that
$k'\le g(k)$ and: $(x,k)\in Q\Longleftrightarrow (x',k')\in Q'$.
If there is a computable $f$ such that $r((x,k))$ is computable in 
space $O(f(k)+\log|x|)$ (on a Turing machine with write-only output tape), then we speak of a $\pl${\em -reduction}.

In descriptive complexity one considers  classical problems as isomorphism closed classes of 
(finite) structures of some fixed vocabulary. 
In the parameterized setting we are led to consider the slices of parameterized problems as such classes of structures. 

\begin{definition}  \label{def:paramprobl}
A {\em parameterized problem} is a subset $Q\subseteq\str\times\nats$ such that for every $k\in\nats$ there is 
a vocabulary $\tau_k$ such that the {\em $k$-th slice of $Q$}, i.e.\ $\{\rela\mid (\rela,k)\in Q\}$, is an isomorphism closed 
class of $\tau_k$-structures.\footnote{ This slightly deviates 
from \cite{FlumGrohe03-describing}, where the $\tau_k$'s are assumed to be pairwise equal and only {\em ordered} structures are considered.} 
If there is $r\in\nats$ such that $\ar(R)\le r$ for all $R\in\bigcup_k\tau_k$, we say that $Q$ {\em has bounded arity}.
\end{definition}

This definition is not in conflict with the mode of speech  above if one views binary strings as structures in the usual way. 
Flum and Grohe~\cite{FlumGrohe01-modelchecking} transferred capturing results (cf.~\cite[Chapter~7]{EbbinghausFlum95-finitemodeltheory}) 
of classical descriptive complexity 
to the parameterized setting via the concept of {\em slicewise definability}.
%
Many parameterized classes could be characterized this way
\cite{FlumGrohe03-describing,ChenFlumGrohe-alternation}. 
For example, a parameterized problem $Q$ is
{\em slicewise FO-definable} if there exists a computable function $d$ mapping every $k\in\nats$ to a first-order sentence $d(k)$ 
defining the $k$-th slice of $Q$ (cf.~\cite{FlumGrohe01-modelchecking}). 

\subsection{Reductions that are quantifier-free after a pre-computation}\label{sec:qfap}

Central to descriptive complexity are first-order reductions which take a structure $\rela$ to the structure $I(\rela)$ where $I$ is a 
first-order interpretation (see e.g.~\cite[Chapter~12.3]{EbbinghausFlum95-finitemodeltheory}). We recall the 
definition (see e.g.~\cite[Chapter~11.2]{EbbinghausFlum95-finitemodeltheory}). 

\begin{definition} \label{df:interpetation} Let $\sigma,\tau$ be (finite, relational) 
vocabularies and $U$ be a unary relation symbol outside~$\tau$.
An {\em interpretation (of~$\tau$ in $\sigma$)} 
is a sequence $I=(\varphi_R)_{R\in \tau\dot\cup\{U,=\}}$ of $\sigma$-formulas such 
that there exists $w\in\nats$ such that for all $R\in\tau\dot\cup\{U\}$ we have 
$\varphi_R=\varphi_R(\bar x_1,\ldots,\bar x_{\ar(R)})$ and $\varphi_==\varphi_=(\bar x_1,\bar x_2)$ 
where every~$\bar x_i$ is a tuple of $w$ variables. The number $w$ is 
the {\em dimension} of $I$. The vocabularies $\sigma$ and $\tau$ are the {\em input} and {\em output vocabulary of}~$I$, respectively.  
An interpretation is {\em quantifier-free} if all its formulas are.
An interpretation $I$
determines the partial function from $\str[\sigma]$ into 
$\str[\tau]$ which maps a $\sigma$-structure $\rela$ to  a $\tau$-structure $\relb$ if there exists
a surjection $f:\varphi_U(\rela)\to B$ such that for all $R\in \tau$ and all $\bar a_1,\bar a_2,\ldots\in\varphi_U(\rela)$:
\begin{enumerate}
\item[] $\rela\models \varphi_{=}(\bar a_1,\bar a_2)\Longleftrightarrow f(\bar a_1)=f(\bar a_2)$;
\item[] $\rela\models \varphi_{R}(\bar a_1,\ldots, \bar a_{\ar(R)})\Longleftrightarrow f(\bar a_1)\cdots f(\bar a_{\ar(R)}) \in R^{\relb}$;
\end{enumerate} 
such a $\relb$, if it exists, is unique up to isomorphism; if no such~$\relb$ exists, the partial function determined by $I$ is
not defined on~$\rela$. 

For technical reasons we extend this partial function to 
a partial function from $\str[\sigma]\cup\{\emptyset\}$ to $\str[\tau]\cup\{\emptyset\}$ 
by adding to its domain $\emptyset$ as well as those $\rela\in\str[\sigma]$ with $\varphi_U(\rela)=\emptyset$; these 
additional arguments are all mapped to $\emptyset$. We denote the resulting  partial function again by $I$.
\end{definition}

We need to agree upon a way how to consider pairs of structures as a single structure:

\begin{definition}  Given a pair $(\rela,\relb)$ of a $\sigma$-structure  $\rela$ and  a $\tau$-structure $\relb$, define the 
structure $\langle\rela,\relb\rangle$ by taking the disjoint union of $\rela$ and $\relb$ and interpreting two new unary 
relation symbols $P_1$ and $P_2$ by the (copies of the) universes of $\rela$ and $\relb$ respectively.
Naturally here, the disjoint union of $\rela$ and $\relb$ has
universe $(\{1\}\times A)\dot\cup(\{2\}\times B)$ and interprets $R\in\sigma\cup\tau$ by
$R_A\cup R_B$ where $R_A:=\emptyset$ if $R\notin\sigma$ and else 
$R_A:=\{((1,a_1),\ldots,(1,a_{\ar(R)}))\mid \bar a\in R^{\rela}\}$; $R_B$ is defined analogously.
For $k\ge 3$ many structures $\rela_1,\ldots,\rela_k$ we inductively set 
$\langle\rela_1,\ldots,\rela_k\rangle:=\langle\langle\rela_1,\ldots,\rela_{k-1}\rangle, \rela_k\rangle$.
\end{definition}

%

It is well-known that NP contains problems that are complete under quantifier-free reductions, 
i.e. reductions computed by a quantifier-free interpretation $I$ as above. 
Dawar and He~\cite{DawarHe09-logicalreductions} transferred 
the notions to the parameterized setting and asked whether central completeness
 results for the classes of the W-hierarchy exhibit a similar robustness. 
More precisely, Dawar and He defined a parameterized reduction $r$ from $Q$ to $Q'$ to be
 {\em slicewise quantifier-free definable} if there 
exists $w\in \nats$ and a computable function $d$ that maps every $k\in\nats$ to some 
quantifier-free interpretation
$d(k)$ of dimension $w$
such that $r((\rela,k))=d(k)((\rela,k))$; here, one views $(\rela,k)$ 
in some suitable way as a single structure.

\begin{definition} \label{df:qfred} 
Let $Q,Q'$ be parameterized problems (Definition~\ref{def:paramprobl}). 
For $k\in\nats$ let~$\tau_k$ be the vocabulary of the $k$-th slice of~$Q$.
A parameterized reduction $r$ from $Q$ to $Q'$
is {\em quantifier-free after a pre-computation} if there are $w\in \nats$ and  computable functions 
\begin{enumerate}\itemsep=0pt
\item[--] $p:\nats \to\nats$ 
\item[--] $a:\nats\to\str$
\item[--] $d$ mapping $k\in\nats$ to a quantifier-free interpretation $d(k)$ of dimension $w$,
\end{enumerate}
such that for all $(\rela,k)\in\str[\tau_k]\times\nats$:
\begin{enumerate}\itemsep=0pt
 \item[--] $d(k)$ is defined on $\langle a(k),\rela\rangle$, and
\item[--] $r((\rela,k))=(\rela',k')$ for $\rela':=d(k)(\langle a(k),\rela\rangle)$ and $k':=p(k)$.
\end{enumerate}
We write $Q\qfred Q'$ to indicate that such a reduction exists, 
and $Q\qfeq Q'$ to indicate that both $Q\qfred Q'$ and $Q'\qfred Q$.
\end{definition}

\begin{remark}
Note that the new {\em p}arameter $p(k)$ is computed by~$p$ from~$k$ alone, 
$a$ is the pre-computation providing an {\em a}uxiliary structure, and~$d$ provides the {\em d}efinition of the 
new structure~$\rela'$.
%
\end{remark}

\begin{remark}
We allow a reduction $r$ to output $(\emptyset,k')$ for some $k'\in\mathbb N$. 
This is considered to be a ``no" instance of any parameterized problem. 
For example, in the definition above we have $r(\rela,k)=(\emptyset,p(k))$ if 
$\varphi_U(\rela)=\emptyset$ where we write $d(k)=(\varphi_R)_{R\in\{U,=\}\cup\ldots}$.
\end{remark}

 \begin{lemma}\label{lem:trans} 
Let $Q,Q',Q''$ be parameterized problems. 
If  $Q\qfred Q'$ and $Q'\qfred Q''$, then $Q\qfred Q''$.
 \end{lemma}

\begin{proof}
We need some folklore combinatorics concerning first-order interpretations. We give some details in order to be clear about our 
special treatment of $\emptyset$.

\medskip

\noindent{\em Claim 1:} For $i\in\{1,2\}$ there is a dimension 1 quantifier-free interpretation $\textit{Pr}_i$ such that  
$\textit{Pr}_i(\langle \rela_1,\rela_2\rangle)$ is defined and
isomorphic to $\rela_i$ for all structures $\rela_1,\rela_2$.\medskip

We omit the easy proof. 

\medskip

\noindent{\em Claim 2:} Assume $I,J$ are quantifier-free interpretations 
of dimensions $w,w'$ respectively, and
such that the output vocabulary of $J$ contains the input vocabulary of $I$. Then there is 
a quantifier-free interpretation $(I\circ J)$ of dimension $w\cdot w'$ which is defined on
a structure $\rela$ whenever both $J$ is defined on $\rela$ and $I$ is defined on $J(\rela)$, and then
outputs $(I\circ J)(\rela)\cong I(J(\rela))$.

Notationally, we understand here that $\emptyset\cong\emptyset$.

\medskip

\noindent{\em Proof of Claim 2:} Let $I=(\varphi_R)_R$ have dimension $w$ and $J=(\psi_S)_S$ have dimension $w'$. 
Let $\sigma$ be the input and $\tau$ be the output vocabulary of $J$. Associate with each variable~$x_j$ \mo{there has been confusion between $I$ and $J$}
a $w'$-tuple $\bar x_j$ of variables. For every 
$\tau$-formula $\psi=\psi(x_1,x_2,\ldots)$ there is  a $\sigma$-formula $J(\psi)=J(\psi)(\bar x_1,\bar x_2,\ldots)$ 
such that for all $\rela\in\str[\sigma]$ with $J(\rela)$ defined and $\neq\emptyset$  we have
for all  $\bar a_1,\bar a_2,\ldots \in\varphi_U(\rela)\subseteq A^{w'}$:
$$
\rela\models J(\psi) (\bar a_1,\bar a_2,\ldots)\Longleftrightarrow J(\rela)\models\psi(f(\bar a_1),f(\bar a_2),\ldots);
$$
here, $f$ is a surjection from $\varphi_U(\rela)$ onto the universe of $J(\rela)$ witnessing that $J(\rela)$ is defined and $\neq \emptyset$.
The formula $J(\psi)$ is obtained from $\psi$ by replacing atomic subformulas $Sx_{i_1}\cdots x_{i_{\ar(S)}}$ and $x_{i_1}=x_{i_2}$ of $\psi$ by
$\psi_S(\bar x_{i_1},\ldots, \bar x_{i_{\ar(S)}})$ and $\varphi_=(\bar x_{i_1},\bar x_{i_2})$, respectively. 
Then the interpretation $(J(\varphi_R))_R$ is as desired whenever $I(\rela)\neq\emptyset$. To additionally ensure output $\emptyset$ 
whenever $J(\rela)=\emptyset$, replace the formula $J(\varphi_U)(\bar x_1,\ldots,\bar x_{w})$ by
$J(\varphi_U)(\bar x_1,\ldots,\bar x_{w})\wedge\bigwedge_{i\in[w]}\psi_U(\bar x_{i})$.
\hfill$\dashv$\medskip

\medskip

\noindent{\em Claim 3:} Assume $I,J$ are quantifier-free interpretations of dimensions $w,w'$ respectively. Then there is a quantifier-free 
interpretation $\langle I,J\rangle$ of dimension $w+w'+2$ which is defined on a structure $\rela$ whenever both $I$ and $J$ are defined on $\rela$, 
and then outputs $\langle I,J\rangle(\rela)\cong\langle I(\rela), J(\rela) \rangle$.

Notationally, we understand here that $\langle\rela,\relb \rangle=\emptyset$ if $\rela=\emptyset$ or $\relb=\emptyset$.
\medskip

\noindent{\em Proof of Claim 3:} Write $I=(\varphi_R)_{R\in \sigma\cup\{U,=\}}$  and $J=(\psi_R)_{R\in \tau\cup\{U,=\}}$. 
Then $\langle I,J\rangle$ is the interpretation $(\chi_R)_{R\in\sigma\cup\tau\cup \{P_1,P_2\}\cup \{U,=\}}$ 
defined as follows. Let $\bar x_i$ 
range over $w$-tuples and $\bar y_i$ range over $w'$-tuples. Set
\begin{eqnarray*}
\chi_U(\bar x_1\bar y_1uv) &:=&\varphi_U(\bar x_1)\wedge\psi_U(\bar y_1),\\
\chi_=(\bar x_1\bar y_1u_1v_1,\bar x_2\bar y_2u_2v_2) &:=&(u_1=v_1\wedge u_2=v_2\wedge\varphi_=(\bar x_1,\bar x_2))\\
&&\ \vee\ ( u_1\neq v_1\wedge u_2\neq v_2\wedge\psi_=(\bar y_1,\bar y_2)),\\
\chi_{P_1}(\bar x\bar yuv)&:=& u=v,\\
\chi_{P_2}(\bar x\bar yuv)&:=& u\neq v.
\end{eqnarray*}
For $R\in\sigma\cup\tau$ define $\chi_R(\bar x_1\bar y_1u_1v_1,\ldots,\bar x_{\ar(R)}\bar y_{\ar(R)}u_{\ar(R)}v_{\ar(R)})$ by
$$
\textstyle
\big(\varphi_R(\bar x_1, \ldots,\bar x_{\ar(R)})\wedge \bigwedge_{i\in[\ar(R)]} u_i=v_i\big)
\vee \big(\psi_R(\bar y_1,\ldots,\bar y_{\ar(R)})\wedge \bigwedge_{i\in[\ar(R)]} u_i\neq v_i\big),
$$
where $\varphi_R$ and $\psi_R$ are inconsistent formulas if 
$R\in\tau\setminus\sigma$ resp. $R\in\sigma\setminus\tau$.
This interpretation is as desired. In particular,  if $I(\rela)=\emptyset$ or $J(\rela)=\emptyset$, then
$\varphi_U(\rela)=\emptyset$ resp. $\psi_U(\rela)=\emptyset$, and then $\chi_U(\rela)=\emptyset$, and hence 
$\langle I,J\rangle(\rela)=\emptyset$.\hfill$\dashv$\medskip

We now prove the lemma. Let $(p,a,d)$ witness $Q\qfred Q'$ and $(p',a',d')$ witness $Q'\qfred Q''$
respectively. Define 
\begin{enumerate}\itemsep=0pt
 \item[--] $\tilde p(k):=p'(p(k))$;
\item[--] $\tilde a(k):=\langle a(k), a'(p(k))  \rangle$;
\item[--] $\tilde d(k):= d'(p(k))\circ\langle 
\textit{Pr}_2\circ \textit{Pr}_1, d(k)\circ\langle \textit{Pr}_1\circ \textit{Pr}_1,\textit{Pr}_2 \rangle \rangle $.
\end{enumerate}

We claim $(\tilde p,\tilde a,\tilde d)$ witnesses $Q\qfred Q''$. 
By construction the interpretations $I\circ J$ and $\langle I,J\rangle$ are computable from $(I,J)$, and hence $\tilde d$ 
is computable. The dimensions of $d(k)$ and $d'(p(k))$ are constant (independent of $k$), say, $w$ and $w'$ respectively.
Then it is easily checked, that $\tilde d(k)$ has constant dimension $w'\cdot (4w+3)$.

Write $\rela':=d(k)(\langle a(k),\rela\rangle),k':=p(k)$
and $\rela'':=d'(k')\big(\langle a'(k'), \rela'\rangle), k'':=p'(k')$. Then $(\rela,k)\in Q$ if and only if $(\rela'',k'')\in Q''$.
It suffices to show that $\tilde d(\langle\tilde a(k),\rela\rangle)\cong\rela''$.

But $\langle \textit{Pr}_1\circ \textit{Pr}_1,\textit{Pr}_2 \rangle (\langle  \tilde a(k),\rela\rangle)\cong\langle a(k),\rela\rangle$, 
and $(\textit{Pr}_2\circ \textit{Pr}_1)(\langle  \tilde a(k),\rela\rangle)\cong a'(p(k))$.
Hence
\begin{eqnarray*}
&&\tilde d(k)(\langle\tilde a(k),\rela\rangle)\cong d'(k')\big(\langle a'(k'), d(k)(\langle a(k),\rela\rangle)\rangle\big)
\cong d'(k')\big(\langle a'(k'), \rela'\rangle)= \rela'',
\end{eqnarray*}
as was to be shown.
\end{proof}

\begin{lemma}\label{lem:plred}
Let $Q,Q'$ be parameterized problems and assume $Q'$ has bounded arity. If
$Q\qfred Q'$, then $Q\le_\pl Q'$.
\end{lemma}

\begin{proof}
 Let $(p,a,d)$ witness $Q\qfred Q'$. 
We need to explain 
how to compute the output $(\rela',k')$ of the reduction in parameterized logarithmic space; here $k'=p(k)$ and 
$\rela'=d(k)(\langle a(k),\rela\rangle)$.
The computation of $p(k),a(k)$ and $d(k)$ 
requires an amount of space that depends on the parameter $k$ only. We show how to compute
 an isomorphic copy of $\rela'$ 
from (a binary encoding of) $\langle a(k),\rela\rangle$ and $d(k)$ in parameterized logarithmic space. 

First note the following: for every formula $\varphi=\varphi(\bar x)$ from $d(k)$ and  for every
length~$|\bar x|$ tuple $\bar a$ from the universe $\tilde A$ of $\langle a(k),\rela\rangle$ 
one can decide in space $O(|\varphi|\log|\varphi|+\log |\tilde A|)$
whether  $\langle a(k),\rela\rangle\models\varphi(\bar a)$. 
Indeed, if $w$ is the dimension of $d(k)$ and $r$ bounds the arity of~$Q'$, then
$\varphi$ has at most $rw$ many variables $\bar x$, and this is an absolute constant.

We compute a copy of $\rela'$ with universe $[m]$ for $m=|A'|$. The binary encoding of the structure
$\langle a(k),\rela\rangle$ determines a linear order on $\tilde A$, and 
this induces a lexicographic order on finite tuples over $\tilde A$. 
To compute $m$, cycle through all $\bar a\in\tilde A^w$ in lexicographic order and increase a counter whenever $\bar a$ passes the following check: 
check that $\langle a(k),\rela\rangle\models \varphi_U(\bar a)$, and check that there is no $\bar a'\in \tilde A^w$ 
lexicographically smaller~$\bar a$ and such that $\langle a(k),\rela\rangle\models(\varphi_=(\bar a,\bar a')\wedge\varphi_U(\bar a'))$.
The latter check is  done by cycling through all $\bar a'\in \tilde A^w$.

Using a similar loop, one can  determine, given $i\in[m]$, the $i$-th tuple $\bar a$ passing the check; 
we denote this tuple by $\bar a_i$. 
Now, to determine the
bits of the encoding of the copy of~$\rela'$, it is sufficient to determine, given a relation symbol $R$ and a tuple
$\bar\imath=(i_1,\ldots,i_{\ar(R)})$ from $ [m]^{\ar(R)}$, whether~$\bar\imath$ satisfies the interpretation of $R$ over $[m]$. This is 
done by computing $\bar a_{i_1},\ldots,\bar a_{i_{\ar(R)}}$ and checking whether 
$\langle a(k),\rela\rangle\models\varphi_R(\bar a_{i_1},\ldots,\bar a_{i_{\ar(R)}})$.
\end{proof}

\paragraph*{Convention} For technical reasons we need to consider 
homomorphism problems $\HOMP{\fancya}$ also for classes $\fancya$ which 
are not necessarily decidable. In such a case we slightly abuse notation and write $\HOMP{\fancya}\qfred Q$ for a 
parameterized problem $Q$ to mean that there are {\em partially} computable 
functions $p,a,d$ whose domain contains $\{\numb{\rela}\mid\rela\in\fancya\}$ such 
that for all $\rela\in \fancya$ and similar $\relb$ we have 
that $d(\numb{\rela})(\langle a(\numb{\rela}),\relb \rangle)=:\relb'$ is defined and:
$$
\rela\hto\relb\quad\Longleftrightarrow\quad ( \relb' , p(\numb{\rela}))\in Q.
$$

\subsection{Homomorphism problems for graph classes}\label{sec:graphclass}

Let $\fG,\fH$ be computably enumerable classes of graphs. In this subsection we show that the associated homomorphism problems 
$\HOMP{\fG^*}$ and $\HOMP{\fH^*}$ are $\qfeq$-equivalent if the graph classes $\fG$ and $\fH$ are $\equiv$-equivalent:

\begin{theorem}\label{theo:dechom} If $\fG \preceq \fH$, then
$\HOMP{\fG^*} \qfred \HOMP{\fH^*}$.
\end{theorem}

\begin{proof}
Choose $w$ such that graphs in  $\fG$ have $\fH$-deconstructions of width at most $w$. Using that $\fH$ is computably  enumerable, it is not hard to see that there is an algorithm that computes given $\relg\in\fG$ a graph $\relh\in\fH$ and an 
$\relh$-deconstruction $(B_h)_{h\in H}$ of $\relg$ of width at most $w$.
Given an instance $(\relb,\numb{\relg^*})$ of $\HOMP{\fG^*}$ with $\relg\in\fG$ and $\relb$ similar to  
$\relg^*$ the reduction outputs $(\relb',\numb{\relh^*})$ where $\relh\in\fH$ is as above and $\relb'$ is 
defined as follows. Assume first that all bags $B_h,h\in H,$ are nonempty.

For $\ell\in\mathbb N$ let  $\textit{PH}(\relg^*,\relb,\ell)$ denote the set of pairs 
$$
(g_1\cdots g_\ell,b_1\cdots b_\ell)\in G^\ell\times B^\ell
$$ 
such that
$\{(g_1,b_1),\ldots,(g_\ell,b_\ell)\}$ is a partial homomorphism from $\relg^*$ to $\relb$.
 For each $h\in H$ choose a tuple $\bar g^h:=g^h_1\cdots g^h_w\in G^w$ that lists the elements of $B_h$. 
Note that $\textit{PH}(\relg^*,\relb,\ell)$ is empty only if $(\relb,\numb{\relg^*})$ is a ``no''-instance of $\HOMP{\fG^*}$. If
$\textit{PH}(\relg^*,\relb,\ell)$ is non-empty it carries a structure $\relb'$ defined as follows:
\begin{eqnarray*}
B'&:=&\textit{PH}(\relg^*,\relb,w),\\
E^{\relb'}&:=&\big\{((\bar g,\bar b),(\bar g',\bar b'))\in \textit{PH}(\relg^*,\relb,w)^2\mid (\bar g\bar g',\bar b\bar b')\in \textit{PH}(\relg^*,\relb,2w)\big\},\\
C^{\relb'}_h&:=&\big\{(\bar g,\bar b)\in \textit{PH}(\relg^*,\relb,w)\mid \bar g=\bar g^h \big\},\quad\textup{ for } h\in H.
\end{eqnarray*}
 
We claim that 
$$
\relg^*\hto\relb\quad\Longleftrightarrow\quad \relh^*\hto\relb'.
$$

If $f$ is a homomorphism from $\relg^*$ to $\relb$, then $h\mapsto(\bar g^h,f(\bar g^h))$ is a homomorphism from~$\relh^*$ to $\relb'$.
Conversely, suppose that $f'$ is a homomorphism from $\relh^*$ to~$\relb'$. 
%
%
If $f'$ maps $h\in H$ to $(\bar g,\bar b)\in B'$, let $f^h$ be the map 
$\{(g_1,b_1),\ldots,(g_w,b_w)\}$. Note that $\dom(f^h)=B_h$ because $f'$ preserves the colours $C_h$.
Any two such maps $f^h$ and $f^{h'}$ are compatible in the 
sense that they agree on arguments on which they are both defined: 
indeed, if $a\in\dom(f^h)\cap\dom(f^{h'})$ then $a\in B_{h}$ and $a\in B_{h'}$, so there is a 
path in $\relh$ from $h$ to $h'$; if $f^h(a)\neq f^{h'}(a)$ then there exists 
neighbors $h_0,h_1$ on this path such that $f^{h_0}(a)\neq f^{h_1}(a)$; then $f^{h_0}\cup f^{h_1}$ is
not a function, and in particular $(f'(h_0),f'(h_1))\notin E^{\relb'}$; 
as $(h_0,h_1)\in E^{\relh}$ this contradicts $f'$ being a homomorphism.
Therefore and since every $g\in G$ appears in some $B_h$, $f:=\bigcup_{h\in H}f^h$ is a
 function from $G$ to $B$. To verify it is a homomorphism we show that it preserves $E$; that is preserves the colours $C_g$ can be seen similary. So, given 
an edge $(g,g')\in E^{\relg}$ we have to show $(f(g),f(g'))\in E^{\relb}$.
Choose $(h,h')\in\refl(E^{\relh})$ such that $\{g,g'\}\in B_h\cup B_{h'}$. Then $f^h\cup f^{h'}$ 
is a 
partial homomorphism from $\relg^*$ to $\relb$: this is clear if $h=h'$; 
otherwise $(h,h')\in E^{\relh}$, so $(f'(h),f'(h'))\in E^{\relb'}$ and 
it follows by definition of $E^{\relb'}$ that $f^h\cup f^{h'}$ is a partial homomorphism. But 
$f^h\cup f^{h'}$ is defined on $g,g'$, so $f^h\cup f^{h'}$ and hence $f$ maps $(g,g')$ 
to an edge in $E^{\relb}$.

We are left to show that there is a quantifier-free interpretation producing $\relb'$ from 
$\langle\tilde\relg,\relb\rangle$ for some structure $\tilde\relg$ computable from~$\relg^*$. 
For $\tilde\relg$ we take the expansion of $\relg^*$ that interprets for every 
$h\in H$ a $w$-ary relation symbol $B_h$ by $\{\bar g^h\}$. 
For $w$-tuples of variables  $\bar x=x_1\cdots x_w$ and $\bar y=y_1\cdots y_w$ consider the formula
\begin{eqnarray*}
 \textstyle 
\textit{ph}^w(\bar x, \bar y)&:=& \textstyle 
\bigwedge_{i\in[w]}  P_1x_i\wedge \bigwedge_{i\in[w]}  P_2y_i
\wedge\bigwedge_{(i,i')\in[w]^2}(x_i=x_{i'}\to y_i=y_{i'}) \\
&& \textstyle \wedge\bigwedge_{(i,i')\in[w]^2} (Ex_ix_{i'}\to Ey_iy_{i'})
\wedge\bigwedge_{i\in[w]}\bigwedge_{g\in G}(C_gx_i\to C_gy_i).
\end{eqnarray*}
Let $\textit{ph}^{2w}$ be similarly defined for $2w$-tuples. Then define
\begin{eqnarray*}
\varphi_U(\bar x\bar y)&:=& \textit{ph}^w(\bar x, \bar y)\\
\varphi_=(\bar x\bar y,\bar x'\bar y')&:=&\textstyle \bigwedge_{i\in[w]}(x_i=x'_{i}\wedge y_i=y'_{i} )\\
\varphi_E(\bar x\bar y, \bar x'\bar y')&:=& \textit{ph}^{2w}(\bar x\bar x', \bar y\bar y')\\
\varphi_{C_h}(\bar x\bar y)&:=&\textstyle B_h\bar x,
\end{eqnarray*}
for $h\in H$. This is a quantifier-free interpretation $I$ of dimension $2w$ such that 
$I(\langle\tilde\relg,\relb\rangle)$ is defined and equals $\emptyset$ if $\textit{PH}(\relg^*,\relb,w)=\emptyset$, and otherwise 
equals $\relb'$. This finishes the 
proof for the case that the bags~$B_h,h\in H,$ are nonempty.

In the general case we can assume that any $B_h$ is nonempty whenever there exists $h'\in H$ in the connected component of $h$ in~$\relh$ 
such that $B_{h'}\neq\emptyset$. Thus we can assume that $\relh$ is the disjoint union of $\relh_0$ and $\relh_1$ such 
that $B_h\neq\emptyset$ for all $h\in H_0$ and $B_h=\emptyset$ for all $h\in H_1$. As seen above we get a $\relb'$ such that
either $\relb'=\emptyset$ and $(\relb,\numb{\relg^*})$ is a ``no''-instance of $\HOMP{\fG^*}$, or
$\relb'$ is a structure such that
\begin{equation*}\label{eq:h0}
\relg^*\hto \relb\quad\Longleftrightarrow \quad\relh^*_0\hto\relb'.
\end{equation*}
Define $\relb''$ as follows using a new vertex $b''\notin B'$: 
\begin{eqnarray*}
B''&:=&B'\cup\{b''\},\\
E^{\relb''}&:=&E^{\relb'}\cup\{(b'',b'')\},\\
C^{\relb''}_h&:=&\left\{
\begin{array}{ll}
C^{\relb'}_h&  , h\in H_0\\
\{b''\}&, h\in H_1
\end{array}\right.;
\end{eqnarray*}
in case $\relb'=\emptyset$ we understand here that the sets $B', E^{\relb'},C^{\relb'}_h$ are empty.

It is straightforward to check that $\relb''$ can be defined by 
a quantifier-free interpretation: e.g. as formula $\varphi_U(\bar x\bar y)$  one may now take 
$$\textstyle
\textit{ph}^w(\bar x, \bar y)\vee\bigwedge_{i\in[w]}(C_{g_0}x_i\wedge x_i=y_i\wedge P_1x_i),
$$
for some fixed vertex $g_0\in G$. Furthermore,
it is easy to see that
\begin{equation*}\label{eq:h0h}
\relh^*_0\hto\relb'\quad\Longleftrightarrow \quad\relh^*\hto\relb'',
\end{equation*}
if $\relb'\neq\emptyset$; if $\relb'=\emptyset$, then $\relh^*\hto\relb''$ fails.
Therefore
 $(\relg^*,\relb)\mapsto (\relh^*,\relb'')$ is a reduction from $\HOMP{\fG^*}$ to $\HOMP{\fH^*}$ as desired. 
\end{proof}
 
\subsection{Homomorphism problems for classes of arbitrary structures}\label{sec:strclass}

Let $\fancya$ be a computably enumerable class of structures.

\begin{theorem}\label{theo:coregraphhom} Suppose $\fancya$ has bounded arity.
Let $\fG$ denote the class of graphs from the
hierarchy~(\ref{eq:hierarchy}) 
having the property that
$\fG\equiv\gr(\core(\fancya))$. Then
$$
\HOMP{\fancya} 
\qfeq\HOMP{\fG^*}.
$$
\end{theorem}

Thus, the complexity of $\HOMP{\fancya}$ is determined in a strong sense by the level which the graph class 
$\gr(\core(\fancya))$ takes in our hierarchy. For example, because the reductions are weaker than fpt-reductions it is the level $\gr(\core(\fancya))$ takes 
in our hierarchy what determines whether $\HOMP{\fancya}$ is W[1]-complete 
or fixed-parameter tractable (cf.~\cite{Grohe07-otherside}), and because it is weaker
 than pl-reductions 
this level determines 
whether $\HOMP{\fancya}$
is in para-L or PATH or TREE (cf.~\cite{ChenMueller13-fineclass-arxiv}).

\medskip

We divide the proof into a sequence of lemmas. The first two of them are analogues of Lemmas~\ref{lemma:polytime-uncore}
and~\ref{lemma:polytime-undo-graph}. 
It is straightforward, and left to the reader, to verify that these reductions are quantifier-free after a pre-computation.


\begin{lemma}\label{lemma:arxive1} If $\fancya$ is a class of cores, then
$\HOMP{\fancya^*} \qfred \HOMP{\fancya}$.
\end{lemma}

%

\begin{lemma}\label{lemma:arxive2} $\HOMP{\gr(\fancya)^*} \qfred \HOMP{\fancya^*}.$
\end{lemma}

\begin{corollary}\label{cor:core} $\HOMP{\fancya} \qfeq \HOMP{\core(\fancya)^*}.$
\end{corollary}

\begin{proof} 
By Lemma~\ref{lem:trans} it suffices to establish the following sequence of reductions:
\begin{eqnarray*}
\HOMP{\fancya}\qfred\HOMP{\core(\fancya)^*}\qfred\HOMP{\core(\fancya)}\qfred\HOMP{\fancya}.
\end{eqnarray*}
To see $\HOMP{\fancya}\qfred\HOMP{\core(\fancya)^*}$, map an instance $(\relb,\numb{\rela})$ with $\rela\in\fancya$ to $(\relb',\numb{\core(\rela)^*})$ where $\relb'$ 
expands $\relb$ interpeting $C_a,a\in A,$ by $C_a^{\relb'}:=B$. It is easy to see that this is quantifier-free after a pre-computation.
That $\HOMP{\core(\fancya)^*}\qfred\HOMP{\core(\fancya)}$ follows from Lemma~\ref{lemma:arxive1}. Finally, to see 
$\HOMP{\core(\fancya)}\qfred\HOMP{\fancya}$ we use (here and only here) that $\fancya$ is computably enumerable: this ensures that there exists 
a partially computable function $p$ such that for all $\rela\in\core(\fancya)$ we have that $p(\numb{\rela})$ is the code of
a structure in $\fancya$ whose core is $\rela$; the reduction 
then maps $(\relb,\numb{\rela})$  to $(\relb,p(\numb{\rela}))$. Again it is clear that this is quantifier-free after a pre-computation.
\end{proof}

We show a partial converse to Lemma~\ref{lemma:arxive2}. This is the technically most
 involved step in the proof of Theorem~\ref{theo:coregraphhom}.

\begin{lemma} \label{lem:classtograph}
If $\fancya$ has bounded arity, then $\HOMP{\fancya^*} \qfred \HOMP{\gr(\fancya)^*}.$
\end{lemma}

\begin{proof} 
We consider two cases depending on whether or not $\gr(\fancya)$ has bounded treewidth.

\medskip

{\em Case 1: }
 Suppose $\gr(\fancya)$ has unbounded treewidth. Then $\L\preceq\gr(\fancya)$ by the 
Hierarchy Theorem. By Theorem~\ref{theo:dechom} we get $\HOMP{\L^*}\qfred \HOMP{\gr(\fancya)^*}$.
Since trivially $\HOMP{\fH^*}\qfred\HOMP{\L^*}$ for every class of graphs $\fH$ it suffices to show
$\HOMP{\fancya^*}\qfred\HOMP{\fH^*}$ for some class of graphs $\fH$. 
In fact, we show that for every $\fancya$ with  bounded arity (not necessarily 
of the form $\fancya^*$)  there exists a class 
of graphs $\fH$ such that $\HOMP{\fancya}\qfred\HOMP{\fH^*}$. 

Let $\sigma$ be a vocabulary and $\rela$ be a $\sigma$-structure.  
We define the following graph $\In(\rela)=(L(\rela)\dot\cup R(\rela),E^{\In(\rela)})$, 
reminiscent of the incidence graph. 
Its universe has ``left'' vertices 
$$
L(\rela):=\big\{(R,\bar a,i)\mid R\in\sigma, \bar a\in R^A, i\in[\ar(R)]\big\}
$$ 
together with ``right'' vertices $R(\rela):=A$; for notational simplicity we assume that $L(\rela)\cap R(\rela)=\emptyset$.
Its edges $E^{\In(\rela)}$ are also devided into two kinds, namely we have edges ``on the left'' 
between $(R,\bar a,i)\in L(\rela)$ and $(R,\bar a,i')\in L(\rela)$
for $i\neq i'$ together with ``left to right'' edges  between $(R,a_1\cdots a_{\ar(R)},i)\in L(\rela)$ and $a_i\in R(\rela)$.

Given $\rela\in\fancya$ and a structure $\relb$, say both of vocabulary $\sigma$, define
the following structure $\relb'$. It expands $\In(\relb)$ to a structure interpreting the language of $\In(\rela)^*$. Namely,
for $(R,\bar a,i)\in L(\rela)$ we set
\begin{eqnarray*}
C_{(R,\bar a,i)}^{\relb'}&:=&\big\{(R,\bar b,i)\mid R\in\sigma,\bar b\in R^{\relb},i\in[\ar(R)]\big\}, 
\end{eqnarray*}
a subset of $L(\relb)$, and for $a\in R(\rela)$ we set 
$
C_{a}^{\relb'}:= R(\relb)=B.
$

We claim that  $(\rela,\relb)\mapsto(\In(\rela)^*,\relb')$ is a reduction as desired.
We first show that 
$$
\rela\hto\relb\quad\Longleftrightarrow \quad \In(\rela)^*\hto\relb'.
$$

To see the forward direction, assume $h$ is a homomorphism from $\rela$ to $\relb$. Define $h':L(\rela)\cup R(\rela)\to B'=L(\relb)\cup R(\relb)$ setting $h'(a):=h(a)\in R(\relb)$ for $a\in R(\rela)$ and $h'((R,\bar a,i)):=(R,h(\bar a),i)\in L(\relb)$ for $(R,\bar a,i)\in L(\rela)$. Note that indeed 
$(R,h(\bar a),i)\in L(\relb)$ because $h(\bar a)\in R^{\relb}$ as $h$ is a homomorphism. It also follows that $h'$ preserves the colours $C_{(R,\bar a,i)}$; that it preserves the colours $C_a$ is clear. It is also clear that it preserves edges ``on the left''. Finally consider a ``left to 
right'' edge between
$(R,a_1\cdots a_{\ar(R)},i)\in L(\rela)$ and $a_i\in R(\rela)$. Then
$(h'((R,\bar a,i)), h'(a_i))=((R,h(\bar a),i),h(a_i))$ and this is in $E^{\relb}= E^{\relb'}$.

Conversely, let $h'$ be a homomorphism from $\In(\rela)^*$ to $\relb'$ and let $h$ be its restriction to $R(\rela)=A$. Since $h'$ preserves the colours $C_a,a\in A,$ the function  $h$ takes values in $R(\relb)=B$. We claim $h$ is a homomorphism from $\rela$ to $\relb$.
Let $\bar a=a_1\cdots a_{\ar(R)}\in R^{\rela}$ for $R\in\sigma$. 
 For each $i\in[\ar(R)]$ there exists a tuple $\bar b_i\in R^{\relb}$ such that $h'((R,\bar a,i))=(R,\bar b_i,i)$ because $h'$
preserves the colour $C_{(R,\bar a,i)}$. But, in fact, $\bar b:=\bar b_i$ does not depend on $i$:
if there would be $i,i'\in[\ar(R)]$ such that $\bar b_i\neq \bar b_{i'}$, then $h'$ would not preserve the edge ``on the left'' between $(R,\bar a,i)$ and $(R,\bar a,i')$. 
It now suffices to show that $\bar b=h(\bar a)$, that is,  $b_i=h(a_i)$ for all $i\in[\ar(R)]$ (where we write $\bar b=b_1\cdots b_{\ar(R)}$). If this would fail for $i\in[\ar(R)]$, then $((R,\bar b,i),h(a_i))=(h'(R,\bar a,i),h'(a_i))\notin E^{\In(\relb)}=E^{\relb'}$ while $((R,\bar a,i),a_i)\in E^{\In(\rela)}$,  contradicting that $h'$ is a homomorphism.

Thus, $(\rela,\relb)\mapsto(\In(\rela)^*,\relb')$ defines a parameterized reduction 
from $\HOMP{\fancya}$ to $\HOMP{\In(\fancya)^*}$ where 
$\In(\fancya):=\{\In(\rela)\mid \rela\in\fancya\}$ is a class of graphs. We are left to show 
that this reduction is quantifier-free after a pre-computation. It is here where we are going to use the assumption that $\fancya$ has bounded arity. 

It suffices to show that some isomorphic copy of
$\relb'$ can be produced from $\langle\rela',\relb\rangle$ by a suitable interpretation $I$ where $\rela'$ is 
some auxiliary structure. It will be clear that $I$ and $\rela'$ are computable from $\numb{\rela}$. 
The dimension of $I$ is $r+1$ where $r\ge 1$ bounds the arity of $\fancya$. If $\rela$ has vocabulary $\sigma$, 
the auxiliary structure $\rela'$ is $(\{0\}\cup(\sigma\times[r]))^*$, i.e. take the $\emptyset$-structure with 
universe $\{0\}\cup(\sigma\times[r])$ and add colours for the elements. 
The interpretation produces the isomorphic copy of $\relb'$ where $(R,\bar b,i)\in L(\relb)$ is replaced by
$((R,i),\bar b\bar 0)$ where $\bar 0$ is a $(r-\ar(R))$-tuple of $0$s. Similarly, $b\in R(\relb)=B$ is replaced
 by $b\bar 0$ with an $r$-tuple of $0$s. 
Recall, $\relb'$ has vocabulary 
$$
\{C_{(R,\bar a,i)}\mid (R,\bar a,i)\in L(\rela)\}\cup\{C_a\mid a\in R(\rela)\}\cup\{E\}.
$$
Letting $\bar x=x_1\cdots x_r$ and $\bar x'=x'_1\cdots x'_r$ be $r$-tuples of variables,  the interpretation $I$  reads as
 follows. For $R\in\sigma,i\in[\ar(R)]$ set
\begin{eqnarray*}
\textit{left}_{R,i}(z\bar x)&:=&\textstyle  P_1z\wedge C_{(R,i)}z\wedge Rx_1\cdots x_{\ar(R)} 
\wedge\textstyle\bigwedge_{j\in[r]\setminus[\ar(R)]}(P_1x_j\wedge C_0x_j),\\
\textit{right}(z\bar x)&:=&\textstyle P_2z\wedge \bigwedge_{j\in[r]}(P_1x_j\wedge C_0x_j),
\end{eqnarray*}
and define
\begin{eqnarray*}
\varphi_U(z \bar x)&:=&\textstyle\bigvee_{R\in\sigma}\bigvee_{i\in [\ar(R)]} \textit{left}_{R,i}(z\bar x)\vee \textit{right}(z\bar x),\\
\varphi_{C_{(R,\bar a,i)}}(z\bar x)&:=&\textit{left}_{R,i}(z\bar x),\\
\varphi_{C_a}(z\bar x)&:=&\textit{right}(z\bar x),\\
\varphi_E(z\bar x,z'\bar x')&:=& \varphi'_{E}(z\bar x,z'\bar x')\vee \varphi'_{E}(z'\bar x',z\bar x),
\end{eqnarray*}
where $\varphi'_{E}(z\bar x,z'\bar x')$ is
\begin{eqnarray*}
\textstyle
\bigvee_{R\in\sigma}\bigvee_{i\in [\ar(R)]}&&\textstyle \Big(\big(\textit{left}_{R,i}(z\bar x)\wedge 
\bigvee_{i'\in[r]\setminus\{i\}}\textit{left}_{R,i'}(z'\bar x')\wedge\bigwedge_{j\in[r]}x_j=x'_j\big)\\
&&\textstyle\ \ \vee\ \big(\textit{left}_{R,i}(z\bar x)\wedge\textit{right}(z'\bar x') \wedge z'=x_i\big)\Big).
\end{eqnarray*}

{\em Case 2: } Suppose $\gr(\fancya)$ has bounded treewidth. Then $\gr(\fancya)\preceq\T$ and by the Hierarchy Theorem there exists a 
class of forests $\fH$ such that $\gr(\fancya)\equiv\fH$. 
Moreover, we can assume that $\fH$ is computably enumerable and closed under deleting proper (connected) components 
since so are the classes of graphs listed in this theorem.
Theorem~\ref{theo:dechom} implies that
 $\HOMP{\gr(\fancya)^*}\qfeq\HOMP{\fH^*}$, so it suffices to show $$\HOMP{\fancya^*}\qfred\HOMP{\fH^*}.$$ 
%
%

Arguing as in Proposition~\ref{prop:decomposition-vs-deconstruction}~(2), 
one sees that there is a $w$ such that for 
every $\rela\in\fancya$ there exists $\relh\in\fH$ such that $\gr(\rela)$ has an $\relh$-decomposition. 
Given an instance $(\relb,\numb{\rela^*})$ of $\HOMP{\fancya^*}$ with $\rela^*\in \fancya^*$ 
compute $\relh\in\fH$ and a $\relh$-decomposition $(B_h)_{h\in H}$ 
of $\gr(\rela)$ of width $w$. As in the proof of Theorem~\ref{theo:dechom} we can assume that each 
component of $\relh$ either has all its bags empty or all its bags nonempty. We can even assume that all bags are nonempty - otherwise
use the above closure assumption on $\fH$ and throw away all  components with empty bags.
Define $\relb'$ as in the proof of Theorem~\ref{theo:dechom}, now with $\rela^*$ replacing $\relg^*$ there. We show
$$
\rela^*\hto\relb\quad\Longleftrightarrow\quad\relh^*\hto\relb'.
$$
The direction from left to right is seen as in the proof of Theorem~\ref{theo:dechom}. Conversely, let $f'$ be a 
homomorphism from $\relh^*$ to $\relb'$ and let $\rela_0$ be a component  of $\rela$. Choose a 
 component, i.e. a tree, $\relh_0$ of $\relh$ such that $(B_{h}\cap A_0)_{h\in H_0}$ is a tree decomposition of $\gr(\rela_0)$.
For $h\in H_0$ define $f_0^h$ as in the proof of Theorem~\ref{theo:dechom}. Again we show that 
$f_0:=\bigcup_{h\in H_0}f_0^h$ is a homomorphism from $\rela_0$ to $\relb$. To see this, let $R$ be a relation symbol
 in the vocabulary of $\rela$ and $\bar a\in R^{\rela_0}$. Then the components of $\bar a$ form a clique in $\gr(\rela_0)$. 
It is well-known for tree decompositions, that cliques are contained in some bag (see e.g. \cite[Lemma~4]{Bodlaender98-arboretum}),
that is, there is $h\in H_0$ such that $B_h$ 
contains all components of $\bar a$. As $f_0^h$ is a partial 
homomorphism from $\rela_0$ to $\relb$ and contains $\bar a$ in its domain, we get $f_0^h(\bar a)=f_0(\bar a)\in R^{\relb}$ as desired. 

That $(\rela^*,\relb)\mapsto(\relh^*,\relb')$ is quantifier-free after a pre-computation is seen as in the proof of Theorem~\ref{theo:dechom}.
\end{proof}

The last two lemmas imply:

\begin{corollary}\label{cor:graph} If $\fancya$ has bounded arity, then 
$$
\HOMP{\fancya^*}\qfeq\HOMP{\gr(\fancya)^*}.
$$
\end{corollary}

\begin{proof}[Proof of Theorem~\ref{theo:coregraphhom}] Note that with $\fancya$ also $\core(\fancya)$ is computably enumerable. Thus
\begin{eqnarray*}
\HOMP{\fancya}&\qfeq&\HOMP{\core(\fancya)^*}\qfeq \HOMP{\gr(\core(\fancya))^*}
\end{eqnarray*}
by Corollaries~\ref{cor:core} and \ref{cor:graph}.
Now apply Theorem~\ref{theo:dechom}.
\end{proof}

\section{Pebble games}\label{sect:pebble-games}

Following the motivation given in the introduction,
this section presents pebble
games which can and will be used to solve
the homomorphism problems associated with the  lower levels of the
hierarchy 
(Theorem~\ref{thm:graph-hierarchy}).
We consider pebble games (cf.~\cite[Section~3.3]{EbbinghausFlum95-finitemodeltheory}) played 
by two players, \emph{Spoiler} and \emph{Duplicator} on two similar
structures $\rela,\relb$. 
Let us call a tuple 
$v = (p_1, \ldots, p_r) \in \nats^{r}$ for $r\ge 1$ a \emph{game vector} 
with $r$ {\em rounds} and $\sum_{i\in[r]}p_i$ {\em pebbles}. 

Informally speaking, the game has $r$ rounds; in the $i$th round,
 Spoiler places $\le p_i$ many pebbles on 
elements of $A$ and Duplicator responds placing equally many pebbles on $B$; Duplicator wins if in the end the 
correspondence between pebbled elements is a partial homomorphism from $\rela$ to $\relb$.

Formally, we define a \emph{Duplicator winning strategy for the $v$-game on $(\rela, \relb)$} to be
a sequence $(W_1, \ldots, W_r)$ of sets of partial homomorphisms from $\rela$ to $\relb$ such that:
%
\begin{itemize}\itemsep=0pt

\item For all $S \subseteq A$ with $|S| \leq p_1$,
the set $W_1$ contains a mapping with domain $S$.

\item For all $i\in[r-1]$, $g \in W_i$ and supersets $S$ of $\dom(g)$
 with $|S \setminus \dom(g)| \leq p_{i+1}$, there is an extension $g'\in W_{i+1}$ of $g$
with domain $S$.

\end{itemize}
We write $\rela \vto \relb$ to indicate that such a 
strategy exists.

This section's first main theorem provides 
a decidable characterization of the homomorphism
problems arising from a single structure
that are solved by the $v$-game.
We say that {\em the $v$-game solves}
$\HOMP{\fancya}$ if for any instance $(\rela, \relb)$ thereof,
the existence of a Duplicator winning strategy for the $v$-game
(on the instance)
 implies that there is a homomorphism from $\rela$
to $\relb$.

Let $v = (p_1, \ldots, p_r)$ be a game vector.
A {\em $v$-decomposition} of a graph $\relg$ is 
an $\relh$-decomposition $(B_h)_{h \in H}$ such that:
\begin{itemize}\itemsep=0pt
\item $\relh$ is a rooted forest of height $< r$, 
\item $|B_h|  \leq p_1$ for all nodes $h$ of $\relh$ at level $1$,
\item $B_h \subseteq B_{h'}$ and $|B_{h'} \setminus B_h| \leq p_{i+1}$ for all 
nodes $h, h'$ of $\relh$ at levels $i$ and
  $i+1$, respectively, with $(h, h') \in E^{\relh}$.
\end{itemize}
The \emph{level} of a node $h$ in $\relh$ is the
number of vertices in the unique path from the root (of $h$'s component) to $h$.
So, roots are considered to be at level $1$.

\begin{theorem}\label{thm:char-v-game-solves}
Let $\rela$ be a structure, and $v$ be a game
vector.  The following are equivalent.
\begin{enumerate}\itemsep=0pt
\item The $v$-game solves $\HOMP{\{\rela\}}$.
\item The graph $\gr(\core(\rela))$ has a $v$-decomposition.
\end{enumerate}
\end{theorem}

This theorem immediately implies that, given a 
structure $\rela$ and a game vector $v$,
it is decidable whether or not 
the $v$-game solves $\HOMP{\{\rela\}}$.
Another immediate consequence of this theorem
is that for any class of structures $\rela$,
the $v$-game solves $\HOMP{\fancya}$
if and only if each graph of the form
$\gr(\core(\rela))$, with $\rela \in \fancya$,
has a $v$-decomposition.

We will provide the proof of this theorem after
presenting two further theorems built on it.
These two theorems 
analyze two natural measures associated with our pebble games:
number of pebbles and number of rounds.  
We show that these two measures correspond exactly ($\pm 1$)
to tree depth and stack depth; this is made precise as follows.

\begin{theorem}[Correspondence between number of pebbles and tree depth]\label{thm:pebbles-td}

Let $\fancya$ be a class of structures, let $\fG$ denote $\gr(\core(\fancya))$, and let $n \geq 1$.
The following are equivalent.
\begin{enumerate}\itemsep=0pt
\item There exists a game vector $v$ with $n$ pebbles such that the $v$-game solves $\HOMP{\fancya}$.
\item The $\underbrace{(1, \ldots, 1)}_{n\;\mathrm{times}}$-game solves $\HOMP{\fancya}$.
\item The class $\fG$ has tree depth $< n$.
\end{enumerate}
\end{theorem}

\begin{proof}
(1 $\Leftrightarrow$ 2): The implication from 2 to 1 is immediate.
To prove that 1 implies~2,
let $v = (p_1, \ldots, p_r)$ be a game vector with $\sum_{i\in[r]} p_i=n$; 
it suffices to show that if there is a 
Duplicator winning strategy $W_1, \ldots, W_n$ for the $(1, \ldots, 1)$-game,
then there is a 
Duplicator winning strategy for the $v$-game.
The sequence
$W_{p_1}, W_{p_1 + p_2}, \ldots, W_{p_1 + \cdots + p_r}$
is straightforwardly verified to be a Duplicator winning strategy
for the $v$-game.

(2 $\Leftrightarrow$ 3): 
By appeal to Theorem~\ref{thm:char-v-game-solves},
it suffices to show that each graph in $\fG$ has a 
$(1, \ldots, 1)$-decomposition if and only if condition 3 holds.
For the forward direction, let $(B_h)_{h \in H}$ be a
$(1, \ldots, 1)$-decomposition 
of a graph $\relg \in \fG$
with respect to the rooted forest~$\relh$.
We process $\relh$
by removing any vertex $h \in H$ with $B_h = \emptyset$ and
by contracting together adjacent vertices $h, h' \in H$ with
$B_h = B_{h'}$,
We obtain 
that each root of $\relh$ has $|B_h| = 1$
and that each node $h'$ having a parent $h$ satisfies
$B_{h'} \supseteq B_h$ and $|B_{h'}| = |B_h| + 1$.
Rename each root $h' \in H$ by the unique element in $B_{h'}$
and
rename each $h' \in H$ with a parent $h$ by the unique element that is
in $B_{h'} \setminus B_h$; then, it holds that the new $\relh$ 
has height $< n$, and 
witnesses that the tree depth of
$\relg$
is $< n$.
For the backward direction, let $\relt$ be a rooted tree of height 
$< n$
witnessing that a component $C$ of a graph $\relg \in \fG$
has tree depth $< n$.
Then $C$ has a $(1, \ldots, 1)$-decomposition
given by
$(B_t)_{t \in T}$ defined by 
$B_t = \{ a \mid  a \textup{ is an ancestor of $t$} \}$;
this is straightforwardly verified.
Combining the $(1, \ldots, 1)$-decompositions
of the components of $\relg$, we obtain the desired decomposition
of $\relg$.
\end{proof}

\begin{theorem}[Correspondence between number of rounds and stack depth]\label{thm:pebbles-sd}

Let $\fancya$ be a class of structures, let $\fG$ denote
$\gr(\core(\fancya))$, and let $r \geq 0$.
The following are equivalent.
\begin{enumerate}\itemsep=0pt
\item There exists a game vector $v$ with $r+1$ rounds
such that the $v$-game solves $\HOMP{\fancya}$.
\item $\fG \preceq \F_r$.
\item The class $\fG$ has bounded tree depth and 
has stack depth $\leq r$.
\end{enumerate}
\end{theorem}

\begin{proof}
(1 $\Leftrightarrow$ 2):
By appeal to Theorem~\ref{thm:char-v-game-solves},
it suffices to show that each graph in $\fG$ has a $v$-decomposition
if and only if $\fG \preceq \F_r$.
For the forward direction, 
let $v=(p_1, \ldots, p_{r+1})$.
By definition, a $v$-decom\-po\-si\-tion is
a $\relh$-decomposition for a rooted forest $\relh$ of height $\leq r$.
Also, it follows from the definition that each bag $B_h$ has 
size $|B_h| \leq p_1 + \cdots + p_r$.
Thus
$\fG$ has $\F_r$-decompositions of bounded width, and
by Proposition~\ref{prop:decomposition-vs-deconstruction},
we obtain $\fG \preceq \F_r$.

For the backward direction, 
by Proposition~\ref{prop:decomposition-vs-deconstruction},
there exists $w \geq 1$ such that
$\fG$ has $\F_r$-decompositions of width $< w$.
When $\relh \in \F_r$ and $(B_h)_{h \in H}$ is a $\relh$-decomposition
of width $< w$ for a graph $\relg$, define $(B'_h)_{h \in H}$
by $B'_h = \bigcup_{a} B_a$ where the union is over all ancestors $a$ of
$h$.
It is readily seen that $(B'_h)_{h \in H}$ is a
$\underbrace{(w, \ldots, w)}_{r+1}$-decomposition of $\relg$.

(2 $\Leftrightarrow$ 3): 
By Proposition~\ref{prop:bounded-characterizations},
whenever condition 2 holds,
$\fG$ has bounded tree depth.
Hence, the desired equivalence 
follows
 from Lemma~\ref{lemma:stack-depth-placement}.
\end{proof}

\subsection{Proof of Theorem~\ref{thm:char-v-game-solves}}

To prove Theorem~\ref{thm:char-v-game-solves} we need to develop 
the theory of the introduced pebble games.

\begin{prop}
\label{prop:transitivity-of-strategy}
The $\vto$ relation is transitive, that is,
when $\rela$, $\relb$, and $\relc$ are similar structures,
if $\rela \vto \relb$ and $\relb \vto \relc$, then $\rela \vto \relc$.
\end{prop}

\begin{proof}
Let $W_1, \ldots, W_r$ be a Duplicator winning strategy 
for the $v$-game on $(\rela, \relb)$,
and let $V_1, \ldots, V_r$ be a Duplicator winning strategy 
for the $v$-game on $(\relb, \relc)$.
For each $i \in [r]$, define $U_i$ to be the set
$\{ g\circ f  \mid f \in W_i, g \in V_i, \dom(g) = \im(f) \}$.
It is straightforward to verify that
$U_1, \ldots, U_r$ is a Duplicator winning strategy 
for the $v$-game on $(\rela, \relc)$.
\end{proof}

\begin{prop}
\label{prop:hom-gives-strategy}
If $\rela$ and $\relb$ are structures such that there exists a
homomorphism $h$ from $\rela$ to $\relb$,
then for any game vector~$v$, there is a Duplicator winning strategy
for the $v$-game on $(\rela, \relb)$.
\end{prop}

\begin{proof}
Such a strategy is given by taking
$W_i$ to be the set containing each restriction of $h$ to 
a subset $S \subseteq A$ with $|S| \leq p_1 + \cdots + p_i$.
\end{proof}

\begin{theorem}
\label{thm:strategy-to-hom}
Let $v = (p_1, \ldots, p_r)$
be a game vector, 
let $\relt, \relb$ be similar structures, and assume that
$\gr(\relt)$ has a $v$-decomposition.
If there exists a Duplicator winning strategy for the $v$-game on
$(\relt, \relb)$, then there is a homomorphism from $\relt$ to $\relb$.
\end{theorem}

\begin{proof}
Let $\relh$ be a forest such that $(B_h)_{h \in H}$ gives a
$v$-decomposition of $\gr(\relt)$.
For each $h \in H$, we define a map $f_h$ 
from $B_h$ to $\relb$ that is a partial homomorphism from $\relt$ to
$\relb$,
in the following inductive manner.
Let $W_1, \ldots, W_r$ be a Duplicator winning strategy for the $v$-game.
For each root $h_0$ of $\relh$, define
$f_{h_0}$ to be a map in $W_1$ that is defined on $B_{h_0}$.
When node $h'$ is the child (in $\relh$) of a node $h$ at level $i$ 
having $f_h$ defined,
we define $f_{h'}$ to be a map in $W_{i+1}$ that is defined on
$B_{h'}$
and that extends $f_h$; such a map exists
by the definition of Duplicator winning strategy
and
by the definition of
$v$-decomposition.
For every element $t \in T$, by the definition of
$\relh$-decomposition,
there exists an $h \in H$ such that $t \in \dom(f_h)$.
If for an element $t \in T$ it holds that
$t \in \dom(f_h) \cap \dom(f_{h'})$
where $h$ is the parent of $h'$ in $\relh$,
by the way in which we defined the mappings $f_h$, $f_{h'}$,
it holds that $f_h(t) = f_{h'}(t)$.
It follows by the connectivity condition of a decomposition 
that for any $h, h' \in H$ such that  $t \in \dom(f_h) \cap
\dom(f_{h'})$,
one has $f_h(t) = f_{h'}(t)$.
Thus $f:=\bigcup_{h \in H} f_h$ is a map from $T$ to $B$, and in fact a homomorphism from $\relt$ to $\relb$:
for any tuple of a relation of $\relt$, its entries are contained in a bag
$B_h$,
and $f$ extends $f_h$,
which is a partial homomorphism from $\relt$ to $\relb$ defined on~$B_h$.
\end{proof}

For each structure $\rela$ and each game vector $v = (p_1, \ldots, p_r)$,
we define the structure $\relt_v(\rela)$ as follows.
Let us say that 
a tuple $(C_1, \ldots, C_m)$ is a \emph{set vector}
(of $v$ in $A$) 
if $1 \leq m \leq r$, 
$C_1 \subseteq \cdots \subseteq C_m\subseteq A$, $|C_1| \leq p_1$, and for each
$i \in [m-1]$, it holds that $|C_{i+1} \setminus C_i| \leq p_{i+1}$.
Let $S=S(v,A)$ be the set of all set vectors of $v$ in $A$.
For a set vector $s = (C_1, \ldots, C_m)$,
when $a \in C_m$, define $u(a,s)$ to be the
smallest prefix of $s$ with $a \in C_{|u(a,s)|}$
(equivalently, with $a \in C_{|u(a,s)|} \cap \cdots \cap C_m$),
and
define $B_s$ to be the set
$\{ (a, u(a,s)) \mid  a \in C_m \}$.
The universe of $\relt_v(\rela)$ is 
$\bigcup_{s \in S} B_s$, and every
symbol $R$ from the vocabulary of $\rela$ is interpreted by 

$$\bigcup_{s\in S}\Big\{ ((a_1, u_1), \ldots, (a_{\ar(R)}, u_{\ar(R)})) \in B_s^{\ar(R)}\mid 
(a_1, \ldots, a_{\ar(R)}) \in R^{\rela} \Big\}.$$


\begin{prop}
\label{prop:tva-has-v-decomp}
For every game vector $v$ and every structure~$\rela$, the graph  $\gr(\relt_v(\rela))$ has a $v$-decomposition.
\end{prop}

\begin{proof}
We use the notation in the definition of $\relt_v(\rela)$ above. One can naturally define a forest with vertices $S$
(the set of all set vectors of $v$ in $A$) by making two vectors $s, s' \in S$ adjacent if and only if $s'$
extends $s$ by one entry (or vice-versa). Taking as roots 
the length 1 set vectors, this gives a rooted forest $\mathcal S$.
We claim that with respect to this rooted forest, $(B_s)_{s \in S}$ is a $v$-decomposition of $\gr(\relt_v(\rela))$.

Let $((a,u),(a',u'))\in\refl(E^{\gr(\relt_v(\rela))})$. If $(a,u)=(a',u')$, then $\{(a,u),(a',u')\}\in B_{u}$. 
If $((a,u),(a',u'))\in E^{\gr(\relt_v(\rela))}$, then there are 
$s\in S$, $(a_1,\ldots, a_{\ar(R)})\in R^{\rela}$ and $i,j\in[\ar(R)]$ such that $(a,u)=(a_i,u(a_i,s))$ and
$(a',u')=(a_j,u(a_j,s))$. Assume $|u(a_i,s)|\le |u(a_j,s)|$ (the case $|u(a_j,s)|\le |u(a_i,s)|$ is symmetric). Then
$a\in C_{|u|}\subseteq C_{|u'|}\ni a'$, so $\{a,a'\}\in C_{|u'|}$ and hence $(a,u),(a',u')\in B_{u'}$. 
To see connectivity of the decomposition, note that every $(a,u)$ appears precisely in those 
$B_s$ such that $u$ is a prefix of $s$, and the set of these $s$ is connected in $\mathcal S$.
Thus $(B_s)_{s \in S}$ is an $\mathcal S$-decomposition of $\gr(\relt_v(\rela))$. 

Further, every root of $\mathcal S$  is a length 1 set vector $s=(C)$, so $|C|=|B_s|\le p_1$. 
If $s$ ands $s'$ are adjacent of levels $i$ and $i+1$ respectively, then $s$ has length $i$ and $s'$ 
extends $s=(C_1,\ldots,C_i)$ by a set $C$. Then
$|B_{s'}\setminus B_{s}|= |C\setminus C_i|\le p_{i+1}$.  Thus $(B_s)_{s \in S}$ is a $v$-decomposition 
of $\relt_v(\rela)$.
%
%
\end{proof}

\begin{prop}
\label{prop:hom-from-tva-to-a}
For each game vector $v$ and each structure $\rela$, there is a homomorphism from $\relt_v(\rela)$ to $\rela$.
\end{prop}

\begin{proof}
The homomorphism is the projection onto the first coordinate,
that is, the mapping that sends an element $(a, u)$ of the universe of $\relt_v(\rela)$
to $a$.
\end{proof}

\begin{theorem}
\label{thm:char-strategy}
Let $\rela, \relb$ be similar structures with vocabulary $\sigma$, and let 
$v = (p_1, \ldots, p_r)$ be a game vector.
The following are equivalent.
\begin{enumerate}\itemsep=0pt

\item There exists a Duplicator winning strategy for the $v$-game 
on $(\rela, \relb)$.

\item 
For every $\sigma$-structure $\relt$ such that $\gr(\relt)$ has a $v$-decomposition:
if there is a homomorphism from $\relt$ to $\rela$,
then there is a homomorphism from $\relt$ to $\relb$.

\item There exists a homomorphism from $\relt_v(\rela)$ to $\relb$.

\end{enumerate}
\end{theorem}

\begin{proof}
(1 $\Rightarrow$ 2): Suppose that  $\rela \vto \relb$ and $\relt \hto \rela$ and that $\gr(\relt)$ has a $v$-decomposition.
We have to show $\relt\hto\relb$.
By Proposition~\ref{prop:hom-gives-strategy},
we have $\relt \vto \rela$.
By Proposition~\ref{prop:transitivity-of-strategy},
we have $\relt \vto \relb$.
We obtain $\relt \hto \relb$ by Theorem~\ref{thm:strategy-to-hom}.

(2 $\Rightarrow$ 3):
This is immediate from Propositions~\ref{prop:tva-has-v-decomp}
and~\ref{prop:hom-from-tva-to-a}.

(3 $\Rightarrow$ 1): 
Let $h$ be a homomorphism from $\relt_v(\rela)$ to $\relb$.
For each $i \in [r]$, define
$W_i = \{ h^+_s  \mid s \textup{ is a set vector of length $i$}
\}$
where $h^+_s$ is the map defined on $\pi_1(B_s)$
that takes $a \in \pi_1(B_s)$ to $h(a, u(a,s))$.
By the definition of $\relt_v(\rela)$, each $W_i$ contains only
partial homomorphisms.

We verify that the $W_i$ give a Duplicator winning strategy, as
follows.
First, when $C \subseteq A$ and $|C| \leq p_1$,
we have that $W_1$ contains a map defined on $C$, via the set vector
$s = (C)$.
Next, suppose that $i \in [r-1]$, that $g \in W_i$, and that 
$C$ is a superset of $\dom(g)$ with $|C \setminus \dom(g)| \leq
p_{i+1}$.
By the definition of $W_i$, there exists a set vector $s$ of length $i$
such that $g = h^+_s$.
Let $s'$ be equal to the length $(i+1)$ set vector that extends $s$
with $C$.  
The mapping $h^+_{s'}$ has domain $C$,
and extends $h^+_s = g$
because for each $a \in \dom(g)$, it holds that
$u(a,s) = u(a,s')$.
\end{proof}

Let $\rela$ be a structure, and let $v$ be a game vector. Let $\HOM{\rela}$ denote the classical problem
$\{\relb\mid \rela\hto\relb\}$. We say that {\em the $v$-game solves}
$\HOM{\rela}$ if for every structure $\relb$ similar to $\rela$,
the existence of a Duplicator winning strategy for the $v$-game
on $(\rela, \relb)$ implies that there is a homomorphism from $\rela$
to $\relb$.
For a class of structures $\fancya$,
we say that {\em the $v$-game solves}
$\HOMP{\fancya}$ if, for each $\rela \in \fancya$, the $v$-game solves $\HOM{\rela}$.

\begin{theorem}[Generalization of Theorem~\ref{thm:char-v-game-solves}]
Let $\rela$ be a structure, and let $v = (p_1, \ldots, p_r)$ be a game
vector.  The following are equivalent.
\begin{enumerate}\itemsep=0pt

\item The $v$-game solves $\HOM{\rela}$.

\item There exists a homomorphism from $\rela$ to $\relt_v(\rela)$.

\item The graph $\gr(\core(\rela))$ has a $v$-decomposition.

\end{enumerate}
\end{theorem}

\begin{proof}
(1 $\Rightarrow$ 2): 
By the equivalence of 1 and 3 in Theorem~\ref{thm:char-strategy},
and the trivial fact that $\relt_v(\rela)\hto \relt_v(\rela)$ we have $\rela\vto\relt_v(\rela)$. Then $\rela\hto\relt_v(\rela)$ by the 
assumption that the $v$-game solves $\HOMP{\rela}$.

(2 $\Rightarrow$ 3): If $\rela\hto \relt_v(\rela)$ then there is a homomorphism $h$ from $\core(\rela)$ to $\relt_v(\rela)$. By Proposition~\ref{prop:hom-from-tva-to-a}
$\relt_v(\rela)\hto\rela$, and clearly $\rela\hto\core(\rela)$, so there is a 
a homomorphism $h'$ from $\relt_v(\rela)$ to $\core(\rela)$. It follows that $h'\circ h$ is a homomorphism from $\core(\rela)$ to itself and thus has to be injective. 
Hence $h$ is injective.
By Proposition~\ref{prop:tva-has-v-decomp}, we know that $\relt_v(\rela)$ has a $v$-decomposition $(B_s)_{s \in   S}$.
It is straightforward to verify that $(B'_s)_{s \in S}$ defined by $B'_s = \{ a \in \core(A) \mid  h(a) \in B_s \}$ is 
a $v$-decomposition of $\gr(\core(\rela))$;
here, $\core(A)$ denotes the universe of $\core(\rela)$.

(3 $\Rightarrow$ 1): Suppose $\gr(\core(\rela))$ has a $v$-decomposition and $\rela \vto \relb$; we want to show
that
$\rela \hto \relb$.
Since $\core(\rela) \hto \rela$, we have $\core(\rela) \vto \rela$
(by Proposition~\ref{prop:hom-gives-strategy}).
By the transitivity of $\vto$
(Proposition~\ref{prop:transitivity-of-strategy}),
 we have
$\core(\rela) \vto \relb$, and by 
Theorem~\ref{thm:strategy-to-hom},
we obtain $\core(\rela) \hto \relb$.
As $\rela \hto \core(\rela)$, we conclude that
$\rela \hto \relb$.
\end{proof}

\section{The homomorphism problems in L} \label{sect:logspace}

For a class of structures $\fancya$ consider the classical problem
$$
\HOM{\fancya}:=\big\{(\relb,\rela)\mid \rela\in\fancya\ \&\ \rela\hto\relb\big\}.
$$
In this section we prove the following
characterization of the $\HOM{\fancya}$ problems
solvable in classical logarithmic space.

\begin{theorem}\label{thm:logchar}
Assume that $\textup{PATH}\neq\textup{para-L}$ and 
let $\fancya$ be a class of structures with bounded arity. 
The following are equivalent. 
\begin{enumerate}\itemsep=0pt
\item $\fancya\in\textup{L}$ and $\gr(\core(\fancya))$ has bounded tree depth.
 \item $\HOM{\fancya}\in\textup{L}$.
\end{enumerate}
\end{theorem}

This
characterization is conditional on the hypothesis $\textup{PATH}\neq\textup{para-L}$, an hypothesis from parameterized complexity theory.
The class PATH has been studied in previous
works~\cite{EST12-spacecomplexity,ChenMueller13-fineclass-arxiv} and is defined as follows.

\begin{definition}
 A parameterized problem $Q\subseteq\{0,1\}^*\times\nats$ is in \textup{PATH} if and only if there are a computable $f:\nats\to\nats$ and 
a nondeterministic Turing machine $\mathbb A$ that accepts $Q$ and for all inputs $(x,k)$ and all runs on it uses space $O(f(k)+\log |x|)$ 
and performs at most $O(f(k)\cdot\log|x|)$  nondeterministic steps.
\end{definition}

One can argue that PATH is a natural and important parameterized complexity class, e.g.\  some fundamental  problems which turn out to be complete for PATH under $\le_\pl$.
The above characterization further underlines its importance. 
We refer to \cite{EST12-spacecomplexity,ChenMueller13-fineclass-arxiv,yijiachenmuellermfcs} for more information. Here, 
let us mention the following result.
\begin{proposition}[\cite{ChenMueller13-fineclass-arxiv}]\label{prop:pathcomplete}
 $\HOMP{\P^*}$ is complete for \textup{PATH} under $\le_\pl$.
\end{proposition}
Let us also mention that, as made precise in~\cite{yijiachenmuellermfcs},
the collapse of PATH to para-L would imply that 
Savitch's theorem can be improved.

\begin{lemma}
\label{lem:vgame-in-L}
For every game vector $v$, 
there exists a logarithmic space algorithm that,
given a pair $(\rela, \relb)$ of similar structures,
decides whether $\rela\vto\relb$.
\end{lemma}

\begin{proof}
 Write $v=(p_1,\ldots,p_r)$ and $\ell:=\sum_{i\in[r]} p_i$. For an instance $(\rela,\relb)$ of $\HOM{\fancya}$
consider the following Boolean formula in the variables $X(\bar a,\bar b)$ for $\bar a\in A^{\ell}, b\in B^{\ell}$:
\begin{eqnarray*}
&& \textstyle
\bigwedge_{\bar a_1\in A^{p_1}}\bigvee_{\bar b_1 \in B^{p_1}}\bigwedge_{\bar a_2\in A^{p_2}}\bigvee_{\bar b_2\in B^{p_2}}\cdots
 \bigwedge_{\bar a_r\in A^{p_r}}\bigvee_{\bar b_r\in B^{p_r}} X(\bar a_1\cdots\bar a_r,\bar b_1\cdots \bar b_r).
\end{eqnarray*}
Further, consider the assignment evaluating $X(\bar a_1\cdots\bar a_r,\bar b_1\cdots \bar b_r)$ by 1 or 0
depending on whether $\{(a_i,b_i)\mid i\in[\ell]\}$ is a partial homomorphism from $\rela$ to $\relb$ or not. This assignment satisfies the formula if and 
only if $\rela\vto\relb$. Both the formula and the assignment are computable from $(\rela,\relb)$ in logarithmic space, and so is the truth value.
\end{proof}

\begin{proof}[Proof of Theorem~\ref{thm:logchar}.]
(1 $\Rightarrow$ 2) 
Let $d\ge 1$ bound the tree depth of $\gr(\core(\fancya))$. By Theorem~\ref{thm:pebbles-td}, the 
 $v$-game solves $\HOM{\fancya}$ where $v:=(\underbrace{1,\ldots, 1}_{d\;\mathrm{ times}})$. It follows that
$\HOM{\fancya}=\{(\relb,\rela)\mid\rela\in\fancya\ \&\ \rela\vto\relb\}$. This  is in L by Lemma~\ref{lem:vgame-in-L} and the assumption $\fancya\in\textup{L}$.

(2 $\Rightarrow$ 1) Clearly, (2) implies $\fancya\in\textup{L}$. For contradiction, assume 
$\gr(\core(\fancya))$ has unbounded tree depth. Then $\P\preceq \gr(\core(\fancya))$ by
Lemma~\ref{lemma:placing-G}. By Theorems~\ref{theo:dechom} and~\ref{theo:coregraphhom} (and Lemma~\ref{lem:plred})
we get $\HOMP{\P^*}\le_\pl\HOMP{\fancya}$. But (2) implies $\HOMP{\P^*}\in\textup{para-L}$, and this contradicts 
the assumption $\textup{PATH}\neq\textup{para-L}$ by Proposition~\ref{prop:pathcomplete}.
%
\end{proof}

%
%
%

\section{Model checking existential sentences} \label{sect:existential-sentences}
In this section we
study the complexity of the parameterized model-checking problems associated with sets of (first-order) sentences $\Phi$, namely
$$
\MC{\Phi}:=\big\{(\rela,\numb{\varphi})\mid \varphi\in\Phi\ \&\ \rela\models\varphi\big\}.
$$
Here, $\numb{\varphi}$ is a natural number coding in some straightforward sense the sentence $\varphi$. 
An {\em existential} sentence is one in which the quantifier $\forall$ does not occur and negation symbols $\neg$ appear only in front of atoms. 
A {\em primitive positive} sentence is one built from atoms by means of $\wedge$ and~$\exists$.
For $q,r\in\nats$ let $\Sigma^q_1[r]$ and $\PP{q}{r}$ denote the sets of existential and, respectively, primitive positive sentences of 
quantifier rank at most $q$ where all appearing relation symbols have arity at most $r$. 

\medskip

The goal of this section is to prove the following.

\begin{theorem}\label{theo:modcheck}
Let $q,r\in\nats, q\ge 1,r\ge 2$. Then $$\MC{\Sigma^q_1[r]}\qfeq\MC{\PP{q}{2}}\qfeq\HOMP{\F^*_{q-1}}.$$
\end{theorem}

We devide the proof into several lemmas. 

\begin{lemma}\label{lem:pptof}
Let $q\in\nats, q\ge 1$. Then  $\HOMP{\F^*_{q-1}}\qfred \MC{\PP{q}{2}}.$
\end{lemma}

\begin{proof} The proof proceeds by standard means defining a ``canonical query''~\cite{ChandraMerlin77-optimal}. Details follow. Given a tree $\relt\in\T$ and $r\in T$ 
we define a primitive positive formula $\varphi_{\relt,r}(x)$ such that for every structure $\rela$ similar to $\relt^*$ and every $a\in A$
$$
\rela\models\varphi_{\relt,r}(a)\quad\Longleftrightarrow\quad
\textup{there exists a homomorphism } h\textup{ from }\relt^*\textup{ to }\relb\textup{ with }h(r)=a;
$$
moreover, $\qr(\varphi_{\relt,r})=h$ for $h$ the height of $\relt$ when rooted at $r$. We give the definition by induction on $h$. For $h=0$ the tree
$\relt$ consists of one node $r$, and we set $\varphi_{\relt,r}(x)=C_rx$. Otherwise, let $t_1\ldots, t_\ell$ list the neighbors of $r$ in $\relt$. 
For $i\in[\ell]$ let $\relt_{i}$ denote the connected component of $\langle T\setminus\{r\}\rangle^{\relt}$ \mo{$\{t\}\to \{r\}$}
containing~$t_i$.
Then $\relt_{i}$ rooted at $t_i$ has height at most $h-1$, so $\varphi_{\relt_i,t_i}(y)$ is defined and we can set
$$\textstyle
\varphi_{\relt,r}(x):= C_rx\wedge\bigwedge_{i\in[\ell]} \exists y(Exy\wedge \varphi_{\relt_i,t_i}(y)). 
$$

Given an instance $(\rela,\numb{\relf^*})$ of $\HOMP{\F^*_{q-1}}$, the forest $\relf$ is the disjoint union of, say, $c$ many trees
$\relt_1,\ldots,\relt_c\in\T_{q-1}$. For $i\in[c]$ choose a root $r_i\in T_i$ witnessing that $\relt_i$ has height at most $q-1$ and set
$\psi:=\bigwedge_{i\in[c]}\exists x\varphi_{\relt_i,r_i}(x)$. Then $(\rela,\numb{\relf^*})\mapsto(\rela,\numb{\psi})$ is a reduction
as desired.
\end{proof}

Let $\DPP{q}{r}$ denote the set of disjunctions of sentences from $\PP{q}{r}$.

\begin{lemma}\label{lem:oed}
Let $q,r\in\nats$. Then $\MC{\Sigma^q_1[r]}\qfred\MC{\DPP{q}{r}}.$
\end{lemma}

\begin{proof}
Given an instance $(\rela,\numb{\varphi})$ of $\MC{\Sigma^q_1[r]}$, let $\tau$ denote the vocabulary of~$\varphi$. For each $R\in\tau$ choose a new relation symbol $\overline{R}$ 
of the same arity. The reduction maps $(\rela,\numb{\varphi})$ to $(\rela',\numb{\psi})$ where $\rela'$ 
expands $\rela$ interpreting every new symbol $\overline{R}$ by $A^{\ar(R)}\setminus R^{\rela}$. The sentence $\psi$ is obtained from $\varphi$
by replacing subformulas of the form $\neg R\bar x$ by~$\overline{R}\bar x$, and then moving disjunctions out using the equivalence
$\exists x(\chi\vee\psi)\equiv (\exists x\chi\vee\exists x\psi)$ and de Morgan rules. This preserves the quantifier rank. It is easy to see
that $(\rela,\numb{\varphi})\mapsto (\rela',\numb{\psi})$ is a reduction as desired.
\end{proof}

The following lemma comprises the key step in the proof of Theorem~\ref{theo:modcheck}. It heavily relies on the results 
from Sections~\ref{sec:hierarchy} and~\ref{sec:classification}.

\begin{lemma}\label{lem:canquery}
Let $q,r\in\nats, q\ge 1$. Then $$\MC{\PP{q}{r}}\qfred\HOMP{\F^*_{q-1}}.$$
\end{lemma}

\begin{proof}
Given an instance $(\relb,\numb{\varphi})$ of $\MC{\PP{q}{r}}$ we can assume 
that the existential quantifiers in $\varphi$ quantify pairwise distinct variables.
Following a known construction of \cite{ChandraMerlin77-optimal} define a structure $\rela_\varphi$ interpreting the vocabulary $\tau$ of $\varphi$ as follows. Its universe
is the set of variables of $\varphi$ and a relation $R\in\tau$ is interpreted by those tuples~$\bar x$ such that the atom $R\bar x$ appears in $\varphi$.
Then
$$
\relb\models \varphi\quad\Longleftrightarrow\quad\rela_\varphi\hto\relb.
$$
This shows $\MC{\PP{q}{r}}\qfred \HOMP{\fancya}$ where $\fancya:=\big\{\rela_\varphi\mid \varphi\in \PP{q}{r} \big\}$. Note that $\fancya$ is decidable and of bounded arity. We claim that
$\HOMP{\fancya}\qfred\HOMP{\F^*_{q-1}}$. 

We have $\HOMP{\fancya}\qfred\HOMP{\fancya^*}$ trivially, and  $\HOMP{\fancya^*}\qfred\HOMP{\gr(\fancya)^*}$ by Lemma~\ref{lem:classtograph}, so 
it suffices to show $\HOMP{\gr(\fancya)^*}\qfred\HOMP{\F^*_{q-1}}$. Applying Theorem~\ref{theo:dechom} it suffices to show
$\gr(\fancya)\preceq\F_{q-1}$.

We argue similarly as in~\cite[Theorem~3.12]{ChenMueller13-fineclass-arxiv}.
Define a graph on the universe  of $\rela_\varphi$, i.e.\ the variables of $\varphi$, by putting an edge between $x$ and $y$ if an only if $\exists x$ and 
$\exists y$ are consecutive quantifiers
in the formula tree of $\varphi$. This graph is a forest $\relf_\varphi$ of height at most $q-1$ when we root each component by the first-quantified variable in it, i.e. the 
variable which is closest
to the root in the formula tree of $\varphi$. 

The closure of $\relf_\varphi$ contains $\gr(\rela_\varphi)$: if $(x,y)\in E^{\gr(\rela_\varphi)}$ then there 
exists $R\in\tau$ and a tuple $\bar x$ of variables containing $x$ and $y$ such that $R\bar x$ appears in $\varphi$; since $\varphi$ is a sentence 
every variable in $\bar x$ and especially $x$ and $y$ are quantified in $\varphi$, so $\exists x$ and $\exists y$ appear on some branch in the formula tree of $\varphi$; 
this means there is a path from $x$ to $y$ in $\relf_\varphi$, so $(x,y)$ is in the closure of $\relf_\varphi$.

It follows from  Propositions~\ref{prop:depth-gives-decomposition}
and~\ref{prop:decomposition-vs-deconstruction} that
$\gr(\rela_\varphi)$ has an $\relf_\varphi$-deconstruction of width at most $q-1$. 
\end{proof}

\begin{remark}\label{rem:nonempty}
Recall that in general we allow our reductions $r$ to output $(\emptyset,k)$ for some~$k$. Without loss of 
generality  this does not happen when reducing to a problem of the form
$\HOMP{\fancya^*}$ for some class of structures $\fancya$. 
Namely, assume $r$ on $(\relb,k)$ outputs $(\relb',\numb{\rela^*})$ and $\relb'$ is 
possibly $\emptyset$; change $r$ to output instead of~$\relb'$ the disjoint union of $\relb'$ and one point that 
does not have any of the colours~$C_a,a\in A, $ of $\rela^*$.
\end{remark}

\begin{lemma}\label{lem:mainmc}
Let $q\ge 1,r\ge 2$. Then $\MC{\DPP{q}{r}}\qfred\HOMP{\F^*_{q-1}}.$
\end{lemma}

\begin{proof}
 Let an instance $(\relb,\numb{\varphi_1\vee\cdots\vee\varphi_k})$ of $\MC{\DPP{q}{r}}$ be given, where $k\ge 1$ and $\varphi_i\in\PP{q}{r}$ for all $i\in[k]$.
For $i\in[k]$ let $(\relb_i,\numb{\relf^*_i})$ be the  output of the reduction of the previous lemma. Then
 $\relf_i\in\F_{q-1}$, so $\relf_i$ is a disjoint union of, say, $c_i$ many trees $\relt_{i1},\ldots,\relt_{ic_i}\in\T_{q-1}$. 
By the Remark~\ref{rem:nonempty} we can assume that all $\relb_i$ are structures, i.e.\ different from $\emptyset$.
Then
\begin{eqnarray}\nonumber
\textstyle \relb\models \varphi_1\vee\cdots\vee\varphi_k
&\Longleftrightarrow&\exists i\in[k]:\relb\models\varphi_i\\\nonumber
&\Longleftrightarrow&\exists i\in[k]: \relf^*_i\hto\relb_i\\\nonumber
&\Longleftrightarrow&\exists i\in[k]\ \forall j\in[c_i]: \relt^*_{ij}\hto\relb_i\\\label{eq:skolemtuple}
&\Longleftrightarrow&\textstyle \forall\bar\jmath\in\prod_{i\in[k]}[c_i]\ \exists i\in[k]:\relt^*_{ij_i}\hto\relb_i,
\end{eqnarray}
where we write $\bar\jmath=j_1\cdots j_{k}$. We can assume that the $\relt_{ij}$'s are pairwise disjoint and for each of them a root is chosen witnessing 
that it has height at most $q-1$. Fix some  $\bar\jmath\in\prod_{i\in[k]}[c_i]$. Define the tree $\relt_{\bar\jmath}$ as follows: take the (disjoint) union of the trees
$\relt_{1j_1},\ldots,\relt_{kj_k}$ and then merge their roots to a new node $r$. Then $r$ witnesses that $\relt_{\bar\jmath}\in\T_{q-1}$ and
all $\relt_{ij_i},i\in[k],$ are subtrees of $\relt_{\bar\jmath}$ pairwise intersecting precisely in $r$. We define a structure $\relb^{\bar\jmath}$ such that
\begin{eqnarray}\label{eq:bigt}
\textstyle  \exists i\in[k]:\relt^*_{ij_i}\hto\relb_i&\Longleftrightarrow&\relt^*_{\bar\jmath}\hto\relb^{\bar\jmath}.
\end{eqnarray}
The structure $\relb^{\bar\jmath}$ is the disjoint union of the structures 
$\relb^{\bar\jmath}_{i}, i\in[k],$ defined next. To ensure~(\ref{eq:bigt}) it suffices that  $\relb^{\bar\jmath}_{i}$ satisfies
\begin{eqnarray}\label{eq:bibarj}
 \relt^*_{\bar\jmath}\hto\relb^{\bar\jmath}_{i}\quad\Longleftrightarrow\quad\relt^*_{ij_i}\hto\relb_i.
\end{eqnarray}
Indeed, assuming (\ref{eq:bibarj}) for every $i\in[k]$ we derive (\ref{eq:bigt}) as follows. If $\relt^*_{ij_i}\hto\relb_i$ for some~$i$, then
$ \relt^*_{\bar\jmath}\hto\relb^{\bar\jmath}_{i}$ for such an $i$ by~(\ref{eq:bibarj}) and, clearly, this implies $\relt^*_{\bar\jmath}\hto\relb^{\bar\jmath}$.
 Conversely, assume  $\relt^*_{\bar\jmath}\hto\relb^{\bar\jmath}$. Since  $\relt^*_{\bar\jmath}$ is connected,
 any homomorphism from $\relt^*_{\bar\jmath}$ to 
$\relb^{\bar\jmath}$ has image in some~$\relb^{\bar\jmath}_{i}$, i.e.\ $\relt^*_{\bar\jmath}\hto\relb^{\bar\jmath}_{i}$ for some $i$; for such an $i$
then (\ref{eq:bibarj}) implies $\relt^*_{ij_i}\hto\relb_i$. 

Given $\relb_i$ it is not hard to define $\relb^{\bar\jmath}_{i}$ such that (\ref{eq:bibarj}) is satisfied. 
First forget all interpretations of symbols outside the vocabulary of $\relt^*_{\bar\jmath}$: these are all colours $C_t$ for $t\in F_i\setminus T_{ij_i}$ as 
well as $C_{r_i}$ where  $r_i$ is the root chosen for $\relt_{ij_i}$. To get a structure in the vocabulary of $\relt^*_{\bar\jmath}$, we add interpretations
of $C_t$ for $t\in T_{\bar\jmath}\setminus T_{ij_i}$: all these $C_t$ are interpreted by $C_{r_i}^{\relb_i}$. Finally, we add loops
to the elements in $C_{r_i}^{\relb_i}$, i.e. set
$$
E^{\relb^{\bar\jmath}_{i}}:=E^{\relb_i}\cup\big\{(b,b)\mid b\in C_{r_i}^{\relb_i}\big\}.
$$
Then (\ref{eq:bibarj}) is straightforwardly verified, and thus we know~(\ref{eq:bigt}). 

Finally, let $\relf$ be the disjoint union of the trees $\relt_{\bar\jmath}, \bar\jmath\in\prod_{i\in[k]}[c_i]$, that is, 
make disjoint copies of the
$\relt_{\bar\jmath}$'s and then take the union. For every $\bar\jmath$ every node $t\in T_{\bar\jmath}$ has a copy $t(\bar\jmath)$ in $F$.
We adapt $\relb^{\bar\jmath}$ by renaming $C_t$ by $C_{t(\bar\jmath)}$, more precisely, let $\tilde\relb^{\bar\jmath}$ be the 
structure with universe $B^{\bar\jmath}$ interpreting $E$ by $E^{\relb^{\bar\jmath}}$ and $C_{t(\bar\jmath)}$ by $C^{\relb^{\bar\jmath}}_t$
for $t\in T_{\bar\jmath}$. This structure is the disjoint union of $\tilde\relb^{\bar\jmath}_i, i\in[k],$ where $\tilde\relb^{\bar\jmath}_i$
is obtained by analoguous renaming applied to~$\relb^{\bar\jmath}_i$.

Then the disjoint union $\relb'$ of the
$\tilde\relb^{\bar\jmath}, \bar\jmath\in\prod_{i\in[k]}[c_i],$ is similar to $\relf^*$, and
\begin{eqnarray*}
\textstyle \relb\models \varphi_1\vee\cdots\vee\varphi_k&\Longleftrightarrow&\textstyle\forall\bar\jmath\in\prod_{i\in[k]}[c_i]: \relt^*_{\bar\jmath}\hto\relb^{\bar\jmath}\\
&\Longleftrightarrow& \relf^*\hto\relb'. 
\end{eqnarray*}
The first equivalence follows from (\ref{eq:bigt}) and (\ref{eq:skolemtuple}), and the second  is trivial. 

It is clear that $\numb{\relf^*}$ can be computed from $\numb{ \varphi_1\vee\cdots\vee\varphi_k}$. We show how to get
 a structure isomorphic to $\relb'$ 
by a suitable interpretation from $\langle \rela,\relb\rangle$
where $\rela$ is a suitable auxiliary structure. 
Choose $(p,a,d)$ for the reduction used above, i.e.\ the one producing $(\relb_i,\numb{\relf^*_i})$ from $(\relb,\numb{\varphi_i})$ 
for every $i\in[k]$. 
Choose $w\in\nats$ such that $d$ maps numbers to interpretations of dimension $w$. For $i\in[k]$ write 
$I_i:=d(\numb{\varphi_i})$ and note $I_i(\langle a(\numb{\varphi_i}),\relb\rangle)=\relb_i$.
We define our auxiliary structure to be 
$$
\textstyle \rela:=\big\langle \ a(\numb{\varphi_1}),\ldots, a(\numb{\varphi_k}),\ ([k]\times\prod_{i\in[k]}[c_i])^*\ \big\rangle.
$$
Note $([k]\times\prod_{i\in[k]}[c_i])$ is the structure interpreting the empty language on the universe $[k]\times\prod_{i\in[k]}[c_i]$.
In the following we let $i$ range over $[k]$ and $\bar\jmath$ over $\prod_{i\in[k]}[c_i]$. 

There is a (quantifier-free) interpretation of dimension 1 which produces an isomorphic copy of
$\langle a(\numb{\varphi_i}),\relb \rangle$ from $\langle\rela,\relb\rangle$.
Composing (see Claim 2 in the proof of Lemma~\ref{lem:trans}) with $I_i$ gives an interpretation $I'_i$ of dimension $w$ such 
that $I'_i(\langle\rela, \relb\rangle)\cong\relb_i$.
Further, there is an interpretation of dimension 1 producing $\relb^{\bar\jmath}_{i}$ from $\relb_i$. 
Composing with $I'_i$  gives an
interpretation $I^{\bar\jmath}_{i}$ of dimension $w$
such that $I^{\bar\jmath}_{i}(\langle\rela, \relb\rangle)\cong\relb^{\bar\jmath}_{i}$. As $\relb^{\bar\jmath}_{i}$ 
is obtained by renaming colours
we get $\tilde I^{\bar\jmath}_{i}$ with $\tilde I^{\bar\jmath}_{i}(\langle\rela, \relb\rangle)\cong\tilde\relb^{\bar\jmath}_{i}$.
 The structure we want to produce is the disjoint union of $(\tilde I^{\bar\jmath}_{i}(\langle\rela, \relb\rangle))_{i\bar\jmath}$. 
 We use the following general claim.

\medskip

\noindent{\em Claim:} Let $J_1,\ldots, J_\ell$ be quantifier-free interpretations of dimension $w$. Then
there is a quantifier-free interpretation $J$ of dimension $w+1$ such that $J(\langle([\ell])^*,\rela\rangle)$ is 
defined whenever all $J_i(\rela),i\in[\ell],$ are 
defined and $\neq\emptyset$, and then $J(\langle([\ell])^*,\rela\rangle)$ is isomorphic to the disjoint union of $J_i(\rela),i\in[\ell]$.

\medskip

\noindent{\em Proof of Claim:} Note $J_i\circ\textit{Pr}_2$ is quantifier-free, has dimension $w$ and produces $J_i(\rela)$ 
from $\langle([\ell])^*,\rela\rangle$ (see Claim 1 in the proof of  Lemma~\ref{lem:trans}). Write $(\varphi^j_{R})_R$ for
this interpretation and let $\bar x_i$ range over $w$-tuples of variables. 

Define $J:=(\psi_R)_R$ as follows. 
\begin{eqnarray*}
\psi_U(y_1\bar x_1)&:=&\textstyle \bigvee_{j\in[\ell]}(C_jy_1\wedge P_1(y_1)\wedge \varphi^j_U(\bar x_1)), \\ 
\psi_=(y_1\bar x_1,y_2\bar x_2)&:=&\textstyle \bigvee_{j\in[\ell]}(C_j(y_1)\wedge C_j(y_2)\wedge\varphi^j_=(\bar x_1,\bar x_2)),\\
\psi_R(y_1\bar x_1,\ldots, y_{\ar(R)}\bar x_{\ar(R)})&:=&\textstyle 
\bigvee_{j\in[\ell]}(\varphi^j_R(\bar x_1,\ldots,\bar x_{\ar(R)})\wedge\bigwedge_{i\in[\ar(R)]}C_j(y_i));
\end{eqnarray*}
where we understand that $\varphi^j_R$ is some inconsistent formula if $R$ is not in the output vocabulary of $J_j$.
\hfill$\dashv$\medskip

To finish the proof, the disjoint union of $(\tilde I^{\bar\jmath}_{i}(\langle\rela , \relb\rangle))_{i\bar\jmath}$ is 
produced by first producing $\langle ([k]\times\prod_{i\in[k]}[c_i])^* , \langle\rela,\relb\rangle\rangle$
from $\langle\rela,\relb\rangle$ and composing with an interpretation $J$ chosen 
for the $\tilde I^{\bar\jmath}_{i}$'s according to the claim.
\end{proof}

\begin{proof}[Proof of Theorem~\ref{theo:modcheck}.] Applying Lemmas~\ref{lem:oed}, \ref{lem:mainmc} and \ref{lem:pptof} in row, we get
\begin{eqnarray*}
&& \MC{\Sigma^q_1[r]}\qfred\MC{\DPP{q}{r}} \qfred\HOMP{\F^*_{q-1}}\qfred \MC{\PP{q}{2}}.
\end{eqnarray*}
Noting $\MC{\PP{q}{2}}\subseteq \MC{\Sigma^q_1[r]}$, Theorem~\ref{theo:modcheck} follows.
\end{proof}


\end{document}